\documentclass[11pt]{article}
\usepackage{fullpage}
\usepackage{graphicx}
\usepackage{bbm}
\usepackage{amsmath,amssymb, amsfonts, xspace, amsthm}
\usepackage{color}
\usepackage{natbib}

\newtheorem{Theorem}{Theorem}
\newtheorem{Theorem*}{Theorem}

\newtheorem{Claim*}[Theorem]{Claim}

\newtheorem{CounterExample*}{$\overline{\hbox{\bf Example}}$}
\newtheorem{Definition}[Theorem]{Definition}

\newtheorem{Example*}[Theorem]{Example}

\newtheorem{Intuition*}[Theorem]{Intuition}
\newtheorem{Joke*}[Theorem]{Joke}
\newtheorem{Lemma}[Theorem]{Lemma}
\newtheorem{Lemma*}[Theorem]{Lemma}
\newtheorem{Open problem}[Theorem]{Open problem}

\newtheorem{Question*}[Theorem]{Question}

\def \bSubexa    {\begin{subexa}}

\newcommand{\ignore}[1]{}

\newcommand{\II}{\mathbb{I}} 
\newcommand{\EE}{\mathbb{E}}

\def \cG     {{\cal G}}
\def \cH     {{\cal H}}

\def \cO     {{\cal O}}

\def \cQ     {{\cal Q}}

\def \cS     {{\cal S}}

\def \tcO    {\tilde\cO}

\newcommand{\poi}{{\rm poi}}
\newcommand{\Var}{{\rm Var}}

\newcommand{\ie}{\textit{i.e.,}\xspace}  

\definecolor{light}{gray}{.75}

\def \upto  {{,}\ldots{,}}

\def \ok    {1\upto k}

\def \setok    {\{\ok\}}

\def \Paren#1{{\left({#1}\right)}}

\newcommand{\ed}{\stackrel{\mathrm{def}}{=}}

\def\ignore#1{}

\newcommand{\bi}{\begin{itemize}}
\newcommand{\ei}{\end{itemize}}

\def\orpro{\mathop{\mathchoice
   {\vee\kern-.49em\raise.7ex\hbox{$\cdot$}\kern.4em}
   {\vee\kern-.45em\raise.63ex\hbox{$\cdot$}\kern.2em}
   {\vee\kern-.4em\raise.3ex\hbox{$\cdot$}\kern.1em}
   {\vee\kern-.35em\raise2.2ex\hbox{$\cdot$}\kern.1em}}\limits}

\def\andpro{\mathop{\mathchoice
 {\wedge\kern-.46em\lower.69ex\hbox{$\cdot$}\kern.3em}
 {\wedge\kern-.46em\lower.58ex\hbox{$\cdot$}\kern.25em}
 {\wedge\kern-.38em\lower.5ex\hbox{$\cdot$}\kern.1em}
 {\wedge\kern-.3em\lower.5ex\hbox{$\cdot$}\kern.1em}}\limits}

\def\simge{\mathrel{%
   \rlap{\raise 0.511ex \hbox{$>$}}{\lower 0.511ex \hbox{$\sim$}}}}

\def\simle{\mathrel{
   \rlap{\raise 0.511ex \hbox{$<$}}{\lower 0.511ex \hbox{$\sim$}}}}

\newcommand{\high}{\text{high}}
\newcommand{\low}{\text{low}}
\newcommand{\guess}{\text{guess}}
\newcommand{\lowvalue}{\text{low}}
\newcommand{\highvalue}{\text{high}}

\def \mup1  {{\boldsymbol{\mu}_1}}
\def \mup2  {{\boldsymbol{\mu}_2}}

\def \mup  {{\boldsymbol\mu}}

\newcommand{\norm}[1]{\left|\left|#1\right|\right|}

\newcommand{\ExtAbs}[1]{\ifthenelse{\equal{\version}{ExtAbs}}{#1}{}}
\newcommand{\FullVer}[1]{\ifthenelse{\equal{\version}{FullVer}}{#1}{}}

\def\tcO{\widetilde{\cO}}

\def\Var{{\rm{Var}}}
\def\lV{\left\lvert}
\def\rV{\right\rvert}

\newcommand{\tilOm}{\widetilde\Omega}
\newcommand{\tilTh}{\widetilde\Theta}

\usepackage[backref,colorlinks,citecolor=blue,bookmarks=true]{hyperref}
\usepackage{aliascnt}
\usepackage[numbered]{bookmark}

\title{Faster Algorithms for Testing under Conditional Sampling}
\author{
\begin{tabular}[t]{c@{\extracolsep{5em}}c@{\extracolsep{5em}}c} 
Moein Falahatgar  & Ashkan Jafarpour & Alon Orlitsky \\
\small\texttt{mfalahat@ucsd.edu} & \small\texttt{ashkan@ucsd.edu} & \small\texttt{alon@ucsd.edu}
\end{tabular}
\vspace{2ex} \\
\begin{tabular}[t]{c@{\extracolsep{8em}}c} 
Venkatadheeraj Pichapathi  & Ananda Theertha Suresh\\
\small\texttt{dheerajpv7@gmail.com} & \small\texttt{asuresh@ucsd.edu} 
\end{tabular}
\vspace{2ex}\\
University of California, San Diego
}
\begin{document}

\maketitle

\begin{abstract}
There has been considerable recent interest in 
distribution-tests whose run-time and
sample requirements are sublinear in the domain-size $k$. 
We study two of the most important tests under the
conditional-sampling model where each query
specifies a subset $S$ of the domain, and the
response is a sample drawn from $S$ according
to the underlying distribution.

For identity testing, which asks whether the
underlying distribution equals a specific given
distribution or $\epsilon$-differs from it, 
we reduce the known time and sample complexities
from $\tcO(\epsilon^{-4})$ to $\tcO(\epsilon^{-2})$,
thereby matching the information theoretic lower bound. 
For closeness testing, which asks whether two
distributions underlying observed data sets are
equal or different, we reduce existing complexity
from $\tcO(\epsilon^{-4} \log^5 k)$ to an even
sub-logarithmic $\tcO(\epsilon^{-5} \log \log k)$ thus
providing a better bound to 
 an open problem in Bertinoro Workshop on Sublinear Algorithms~\citep{Previous_authors}.

\end{abstract}


\newcommand{\scikeds}[2]{N_{\textrm{id}}({#1},\epsilon#2)} 
\newcommand{\sciked}{\scikeds{k}{,\delta}}
\newcommand{\sciket}{\scikeds{k}{,0.1}}
\newcommand{\scike}{\scikeds{k}{}}

\newcommand{\scckeds}[2]{N_{\textrm{cl}}({#1},\epsilon#2)} 
\newcommand{\sccked}{\scckeds{k}{,\delta}}
\newcommand{\sccket}{\scckeds{k}{,0.1}}
\newcommand{\sccke}{\scckeds{k}{}}

\newcommand{\sccike}{N^*_{\textrm{id}}(k,\epsilon)}
\newcommand{\scccke}{N^*_{\textrm{cl}}(k,\epsilon)}
\newcommand{\scccko}{N^*_{\textrm{cl}}(k,1/4)}

\section{Introduction}
\subsection{Background}
The question of whether two probability distributions are the same
or substantially different arises in many important applications.
We consider two variations of this problem:
\emph{identity testing} where one distribution is known while the
other is revealed only via its samples,
and 
\emph{closeness testing}
where both distributions are revealed only via their samples.

As its name suggests, identity testing arises when an identity
needs to be verified. For example, testing whether a given
person generated an observed fingerprint, if a specific author wrote
an unattributed document, or if a certain disease caused the symptoms
experienced by a patient. In all these cases we may have sufficient
information to accurately infer the true identity's underlying
distribution, and ask whether this distribution
also generated newly-observed samples.
For example, multiple original high-quality fingerprints
can be used to infer the fingerprint structure, and then be
applied to decide whether it generated newly-observed fingerprints.

Closeness testing arises when we try to discern whether the same
entity generated two different data sets. For example, if two
fingerprints were generated by the same individual, two documents
were written by the same author, or two patients suffer from the
same disease. In these cases, we do not know the distribution
underlying each data set, but would still like to determine whether
they were generated by the same distribution or by two different ones.

Both problems have been studied extensively.
In the hypothesis-testing framework, researchers studied the
asymptotic test error as the number of samples tends to
infinity, ~\citep[see ][and references therein]{Ziv88, Unnikrishnan12}.
We will follow a more recent, non-asymptotic approach.
Two distributions $p$ and $q$ are 
$\epsilon$-\emph{far} if 
\[
\norm{p-q}_1\ge\epsilon.
\]
An \emph{identity test} for a given distribution $p$ considers independent
samples from an unknown distribution $q$ and declares either
$q=p$ or they are $\epsilon$-far.
The test's \emph{error probability} is the highest probability that it errs,
maximized over $q=p$ and every $q$ that is $\epsilon$-far from $p$.
Note if $p$ and $q$ are neither same nor $\epsilon$-far,
namely if $0<\norm{q-p}_1<\epsilon$, neither answer constitutes an error.

Let $\sciked$ be the smallest number of samples to identity test every $k$-element distribution with error
probability $\leq\delta$.
It can be shown that the sample complexity depends on $\delta$ mildly,
$\sciked \leq  \cO(\sciket)\cdot\log\frac1\delta$. Hence we focus on $\sciket$, denoting it by $\scike$.

This formulation was introduced by~\cite{GoldreichR00} 
who, motivated by testing graph expansion, considered
identity testing of uniform distributions.~\cite{Paninski08} showed that the sample complexity of identity testing for the uniform distributions is $\Theta(\epsilon^{-2}\sqrt{k})$. 
General identity testing was studied by~\cite{BatuFFKRW01} who
showed that $\scike\le \tcO(\epsilon^{-2}\sqrt{k})$, and 
recently~\cite{ValiantV13} proved a matching lower bound, implying that
$\scike=\Theta(\epsilon^{-2}\sqrt{k})$, where $\tcO$ 
and later $\tilTh$ and $\tilOm$, hide multiplicative logarithmic factors.

Similarly, a \emph{closeness test} takes independent samples
from $p$ and $q$ and declares them either to be the same or
$\epsilon$-far.
The test's \emph{error probability} is the highest probability that it errs,
maximized over $q=p$ and every $p$ and $q$ that are
$\epsilon$-far.
Let $\sccked$ be the smallest number of samples that suffice to
closeness test every two 
$k$-element 
distributions
with error 
probability $\leq \delta$.
Here too it suffices to consider 
$\sccke\ed\sccket$.

Closeness testing was first studied by~\cite{BatuFRSW00} who showed that
$\sccke \leq \tcO(\epsilon^{-4}k^{2/3})$.
Recently~\cite{Valiant11a,ChanDVV13}
showed that $\sccke=\Theta(\max(\epsilon^{-4/3}k^{2/3}, \epsilon^{-2} \sqrt{k}))$.

\subsection{Alternative models}
The problem's elegance, intrinsic interest, and potential
applications have led several researchers to consider scenarios where
fewer samples may suffice.
Monotone, log-concave, and $m$-modal distributions
were considered in~\cite{RubinfeldS09, DaskalakisDSVV13,  DiakonikolasKN14,ChanDSS14},
and their sample complexity was shown to decline from a polynomial in
$k$ to a polynomial in $\log k$.
For example, identity testing of
monotone distributions over $k$ elements requires
$\cO(\epsilon^{-5/2}\sqrt{\log k})$ samples, and identity testing 
log-concave distributions over $k$ elements requires
$\tcO(\epsilon^{-9/4})$ samples, independent of the support size $k$.

A competitive framework that analyzes the optimality for every
pair of distributions was considered in~\cite{AcharyaDJOPS12, ValiantV13}.
Other related scenarios include classification~\citep{AcharyaDJOPS12}, outlier detection~\citep{AcharyaJOS14b},
testing collections of distributions~\citep{LeviRR13}, testing for the class of monotone distributions~\citep{BatuKR04}, testing for the class of Poisson Binomial distributions~\citep{AcharyaD15},
testing under different distance measures~\citep{GuhaMV09, Waggoner15}.

\ignore{
Some of the connections to the related problems of classification and outlier
detection are outlined in~\cite{AcharyaDJOPS12}
and~\cite{AcharyaJOS14b} respectively. Testing properties of collections of distributions was considered in~\cite{LeviRR13}.
}

Another direction lowered the sample complexity of all distributions
by considering more powerful queries. Perhaps the most natural is the
\emph{conditional-sampling} model introduced independently
in~\cite{ChakrabortyFGM13} and~\cite{CanonneRS14},
where instead of obtaining samples from the entire support set, each query
specifies a \emph{query set} $S\subseteq[k]$ and the samples
are then selected from $S$ in proportion to their original
probability, namely element $i$ is selected with probability
\[
P_S(i)=
\begin{cases}
\frac{p(i)}{p(S)} & i \in S,\\
0 &\text{otherwise,}
\end{cases}
\]
where $p(S)$ is the probability of set $S$ under $p$.
Conditional sampling is a natural extension of sampling,
and~\cite{ChakrabortyFGM13} describes several scenarios where it may arise. Note that unlike other works in distribution testing, conditional sampling algorithms can be adaptive, \ie each query set can depend on previous queries and observed samples. It is similar in spirit to the machine learning's
popular \emph{active testing} paradigm, where additional information is interactively requested for specific domain elements.
\cite{BalcanBBY12} showed that 
various problems such as testing unions of intervals, testing linear separators
benefit significantly from the active testing model.

Let $\sccike$ and $\scccke$ be the number of samples required for
identity- and closeness-testing under conditional sampling model.
For identity testing,~\cite{CanonneRS14} showed that 
conditional sampling eliminates the dependence on $k$,
\[
\Omega(\epsilon^{-2})\le \sccike\le\tcO(\epsilon^{-4}).
\] 
For closeness testing, the same paper showed that
\[
\scccke\le\tcO(\epsilon^{-4} \log^5 k).
\]
\cite{ChakrabortyFGM13} showed that  $\sccike \leq \text{poly}(\log^*k, \epsilon^{-1})$ and designed a $\text{poly}(\log k,\epsilon^{-1})$ algorithm for testing
any label-invariant property. They also derived a $\Omega(\sqrt{\log\log k})$ lower bound for testing any label-invariant property.

An open problem posed by~\cite{Previous_authors} asked the 
sample complexity of closeness testing under conditional sampling which was partly answered by~\cite{AcharyaCK14}, who showed
\[
\scccko \geq \Omega(\sqrt{\log\log k}).
\]
\subsection{New results}
Our first result resolves the sample complexity of 
identity testing with conditional sampling.
For identity testing we show that
\[
\sccike\le\tcO(\epsilon^{-2}).
\]
Along with the information-theoretic lower bound above, this yields
\[
\sccike=\tilTh(\epsilon^{-2}).
\]
For closeness testing, we address the open problem of~\cite{Previous_authors} by
reducing the upper bound from $\log^5 k$ to $\log\log k$.
We show that
\[
\scccke
\le
\tcO\Paren{\epsilon^{-5}\log\log k}.
\]
This very mild, double-logarithmic dependence on the alphabet
size may be the first sub-poly-logarithmic growth rate of any
non-constant-complexity property and together
with the lower bound in~\cite{AcharyaCK14} shows that the 
dependence on $k$ is indeed a poly-double-logarithmic.

Rest of the paper is organized as follows. We first study identity testing in Section~\ref{sec:identity}. In Section~\ref{sec:closeness} we propose an algorithm for closeness testing. All the proofs are given in Appendix.

\section{Identity testing}
\label{sec:identity}
In the following, $p$ is a distribution over $[k]\ed\setok$, 
$p(i)$ is the probability of $i\in[k]$, $|S|$ is the
cardinality of $S\subseteq[k]$, $p_S$ is the conditional
distribution of $p$ when $S$ is queried, and $n$ is the number of samples. 
For an element $i$, $n(i)$ is used to denote the number of occurrences of $i$.

This section is organized as follows. We first motivate our identity test using restricted uniformity testing, 
a special case of identity testing.
We then highlight two important aspects of our identity test: finding a \emph{distinguishing element} $i$ and finding a 
\emph{distinguishing set} $S$. We then provide a simple algorithm for finding a distinguishing element.
As we show, finding distinguishing sets are easy for testing \emph{near-uniform} 
distributions and we give an algorithm for testing near-uniform distributions. 
We later use the near-uniform case as a subroutine for testing any general distribution.
\subsection{Example: restricted uniformity testing}
Consider the class of distributions $\cQ$, 
where each $q \in \cQ$ has $k/2$ elements with probability $(1+\epsilon)/k$, 
and $k/2$ elements with probability $(1-\epsilon)/k$.
Let $p$ be the uniform distribution, namely $p(i)=1/k$ for all $1\leq i \leq k$. Hence for every $q \in \cQ$, $\norm{p-q}_1 = \epsilon$.

We now motivate our test via a simpler \emph{restricted uniformity testing}, a
special case of identity testing where one determines
if a distribution is $p$ or if it belongs to the class $\cQ$.

If we know two elements $i,j$ such that
$q(i) = \frac{1+\epsilon}{k}> \frac{1}{k}=p(i)$ and $q(j) = \frac{1-\epsilon}{k} < \frac{1}{k} = p(j)$, it suffices to 
consider the set $S = \{i,j\}$. For this set 
\[p_S(i) = \frac{p(i)}{p(i)+p(j)} = p_S(j)= \frac{p(j)}{p(i)+p(j)} = 
\frac{1/k}{2/k}= \frac{1}{2},
\] while 
\[
q_S(i) = \frac{q(i)}{q(i)+q(j)} = \frac{(1+\epsilon)/k}{(1+\epsilon)/k + (1-\epsilon)/k} = \frac{1+\epsilon}2,
\]
 and similarly $q_S(j) =(1-\epsilon)/2$. Thus differentiating between $p_S$ and $q_S$ is same as differentiating between $B(1/2)$ and $B((1+\epsilon)/2)$ for which a simple application of the Chernoff bound shows that  $\cO(\epsilon^{-2})$ samples suffice. Thus the sample complexity is $\cO(\epsilon^{-2})$ if we knew such a set $S$.

Next consider the same class of distributions $\cQ$, but without the knowledge of elements $i$ and $j$. We can pick two elements uniformly at random from all possible ${k \choose 2}$ pairs. With probability $\geq 1/2$, 
the two elements will have different probabilities as above, and again we could determine whether root the distribution 
is uniform.
Our success probability is half the success probability when $S$ is known, but it can be increased by repeating the experiment several times and declaring the distribution to be non-uniform if one of the choices of $i$ and $j$ indicates non-uniformity.

While the above example illustrates tests for uniform distribution,
for non-uniform distributions finding elements $i,j$ can be difficult.
Instead of finding pairs of elements, we find a distinguishing element $i$ and a distinguishing set $S$
such that $q(i)<p(i)\approx p(S)<q(S)$, thus when
conditional samples from $S \cup \{i\}$ are observed,
the number of times $i$  appears would differ significantly,
and one can use Chernoff-type arguments to
differentiate between \texttt{same} and \texttt{diff}.
While previous authors have used similar methods, our main contribution is to design a 
information theoretically
near-optimal identity test.

Before we proceed to identity testing, 
we quantify the Chernoff-type arguments formally using~\textsc{Test-equal}. It takes samples from two unknown binary distributions $p,q$ (without loss of generality assume over $\{0,1\}$),  error probability $\delta$, and a parameter $\epsilon$ and
it tests if $p=q$ or $\frac{(p-q)^2}{(p+q)(2-p-q)}\geq \epsilon$.
We use the chi-squared distance $\frac{(p-q)^2}{(p+q)(2-p-q)}$ as the measure of distance instead of $\ell_1$  since it captures the dependence on sample complexity more accurately. For example, consider two scenarios:
$p,q = B(1/2),B(1/2+\epsilon/2)$ or $p,q = B(0),B(\epsilon/2)$. In both cases $\norm{p-q}_1 = \epsilon$, but the number of samples required to distinguish $p$ and $q$ in the first case is  $\cO(\epsilon^{-2})$, while in the second case $\cO(\epsilon^{-1})$ suffice. 
However, chi-squared distance correctly captures 
the sample complexity as in the first case it is $\cO(\epsilon^2)$ 
and in the second case it is $\cO(\epsilon)$.  
While several other simple hypothesis tests exist, the algorithm below has near-optimal sample complexity in terms of $\epsilon,\delta$.
\begin{center}
\fbox{\begin{minipage}{1.0\textwidth}
Algorithm \textsc{Test-equal} \newline
\textbf{Input:} chi-squared bound $\epsilon$, error $\delta$, distributions $B(p)$ and $B(q)$.\newline
\textbf{Parameters:} $n=\cO({1}/{\epsilon})$.
\newline
Repeat $18 \log \frac{1}{\delta}$ times and output the majority:
\begin{enumerate}
\item
Let $n' = \poi(n)$ and $n'' = \poi(n)$ be two independent Poisson variables with mean $n$.
\item
Draw samples $x_1,x_2\ldots x_{n'}$ 
from the first distribution  and 
$y_1,y_2\ldots y_{n''}$ from the second one.
\item
Let $n_1 = \sum_{i=1}^{n'} x_i$ and $n_2 = \sum_{i=1}^{n''} y_i$.
\item 
If 
$\frac{(n_1 -n_2)^2 - n_1-n_2}{n_1+n_2-1} 
+
\frac{(n_1 -n_2)^2 - n_1-n_2}{n'+n'' - n_1-n_2-1} 
\leq \frac{n\epsilon}{2}$
then output \texttt{same}, else \texttt{diff}.
\end{enumerate}
\end{minipage}}
\end{center} 
\begin{Lemma}[Appendix~\ref{app:test_equal}]
\label{lem:test_equal}
If $p=q$, then \textsc{Test-equal} outputs \texttt{same}
with probability $1-\delta$.
If $\frac{(p-q)^2}{(p+q)(2-p-q)} \geq \epsilon$, it outputs \texttt{diff}
with probability $\geq 1- \delta$. 
Furthermore the algorithm uses $\cO \bigl(\frac{1}{\epsilon} \cdot \log \frac{1}{\delta} \bigr)$ samples.
\end{Lemma}
\ignore{
\begin{Lemma}
\label{lem:majority}
Given a value $p$ and $n=11\frac{\log\frac2\epsilon}{p}$ samples
from $\cH_0 = B(p_1)$ or $\cH_1 = B(p_2)$ such that 
$p_1 \leq p \leq p_2/2$, with probability $\geq 1-\epsilon$, one can 
correctly identify the underlying hypothesis.
\end{Lemma}
\begin{proof}
Let $X_1,X_2\ldots X_n$ be the observed samples.
Suppose a test that outputs $\cH_0$ if $\sum_{i=1}^n X_i <\frac32 np$
and $\cH_1$ otherwise.
If Hypothesis $\cH_1$ is true, after using Lemma~\ref{lem:eqchern},
with probability $>1-\epsilon$,
\[
 p_2-\frac{\sum_{i=1}^n X_i}n \leq \sqrt{ \frac{2p_2(1-p_2)}{n} \log \frac{2}{\epsilon}} + \frac{2}{3} \frac{\log \frac{2}{\epsilon}}{n}.
\]
By substituting $n$ in RHS and simplifying the result
we get $\frac{\sum_{i=1}^n X_i}n  > \frac32 p$ hence our 
test declare $\cH_1$ with probability $>1-\epsilon$.
Similarly, if hypothesis $\cH_0$ is true then
$\frac{\sum_{i=1}^n X_i}n  <\frac32 p $ with probability $>1-\epsilon$
and our test declare $\cH_0$. Hence the lemma is proved.
\end{proof}
}
\subsection{Finding a distinguishing element $i$}
We now give an algorithm to find
 an element $i$ such that $p(i) > q(i)$.
In the above mentioned example, we could find such an element with probability $\geq1/2$, by randomly selecting $i$ out of all elements.
However, for some distributions, this probability is much lower.
For example consider the following distributions $p$ and $q$.
$p(1)=\epsilon/2$,
$p(2)=0$,
$p(i)=\frac{1-\epsilon/2}{k-2}$ for $i\ge2$,
and
$q(1)=0$,
$q(2)=\epsilon/2$,
$q(i)=\frac{1-\epsilon/2}{k-2}$ for $i\ge 2$.
Again note that $\norm{p-q}_1 = \epsilon$.
If we pick $i$ at random, the chance that $p(i) > q(i)$ is $1/k$,
very small for our purpose. A better way of selecting $i$ would be sampling according to $p$ itself. For example, the probability of finding an element $i$ such that $p(i) > q(i)$ when sampled from $p$ is $\epsilon/2 \gg 1/k$.

We quantify the above idea next by using the following simple algorithm 
that picks elements such that $p(i) > q(i)$. We first need the following definition. Without loss of generality assume that the elements are ordered such that $p(1) \geq p(2) \geq p(3) \ldots \geq p(k)$.
\begin{Definition}
For a distribution $p$, element $i$ is $\alpha$-heavy, if $
\sum_{i': i' \geq i} p(i') \geq \alpha$.
\end{Definition}
\ignore{\begin{Definition}
For a distribution pair $p,q$, an element $i$ is $\beta$-far, if $p(i) - q(i) \geq \beta p(i)$.
\end{Definition}}
\ignore{
$\beta$-far quantifies how far $p(i)$ and $q(i)$ are. We also need another definition, if $p$ is one of the relatively heavy elements. The usage of this definition becomes clear in the proofs.}

As we show in proofs, symbols that are heavy ($\alpha$ large)
can be used as distinguishing symbols easily and hence our goal is to choose symbols such that $p(i) > q(i)$ and $i$ is $\alpha$-heavy for a large value of $\alpha$. 
To this end, first consider an auxiliary result that shows if for some non-negative values $a_i$, $\sum_{i} p(i) a_i > 0$, then the following sampling algorithm will pick an element $x_i$ such that $x_i$ is $\alpha_i$-heavy and $a_{x_i} \geq \beta_i$. While several other algorithms have similar properties, the following algorithm achieves a good trade-between $\alpha$ and $\beta$ (one of the tuples satisfy $\alpha \beta = \tilde \Omega(1)$), hence
 is useful in achieving near-optimal sample complexity.

\begin{center}
\fbox{\begin{minipage}{1.0\textwidth}
Algorithm \textsc{Find-element} \newline
\textbf{Input:} Parameter $\epsilon$, distribution $p$. \newline
\textbf{Parameters:} $m = 16/\epsilon$, $\beta_j = j\epsilon/8 $, $\alpha_j = 1/(4j \log (16/\epsilon)) $.
\begin{enumerate}
\item Draw $m$ independent samples $x_1,x_2\ldots x_m$ from $p$.
\item Output tuples $(x_1,\beta_1,\alpha_1), (x_2, \beta_2,\alpha_2),\ldots ,(x_m,\beta_m,\alpha_m)$.
\end{enumerate}
\end{minipage}}
\end{center} 
\begin{Lemma}[Appendix~\ref{app:find_far}]
\label{lem:find_far}
For $1 \leq i \leq k$, let $a_i$ be such that $0 \leq a_i \leq 2$. 
If $\sum^k_{i=1} p_ia_i \geq \epsilon/4$, then with probability $\geq 1/5$, at least one tuple $(x,\alpha,\beta)$ returned by \textsc{Find-element}$(\epsilon,p)$ satisfy the property that 
$x$ is $\alpha$-heavy and $a_x \geq \beta$.
Furthermore it uses $16/\epsilon$ samples.
\end{Lemma}
We now use the above lemma to pick elements such that $p(i) > q(i)$. Since $\norm{p-q}_1 \geq \epsilon$, 
\[
\sum_{i : p(i) \geq q(i)} \left( p(i) - q(i) \right) \geq \epsilon/2.
\]
Hence 
\[
\sum_{i} p(i)  \max \left( 0,\frac{p(i) - q(i)}{p(i)} \right) \geq \frac{\epsilon}{2}.
\]
Applying Lemma~\ref{lem:find_far} with $a_i = \max \left( 0,\frac{p(i) - q(i)}{p(i)} \right)$, yields
\begin{Lemma}
\label{lem:pickgood}
If $\norm{p-q}_1 \geq \epsilon$, then with probability $\geq 1/5$ at least one of the tuple  $(i,\beta,\alpha)$ returned by \textsc{Find-element}$(\epsilon,p)$ satisfies $p(i) - q(i) \geq \beta p(i)$ and $i$ is $\alpha$-heavy. Furthermore \textsc{Find-element} uses $16/\epsilon$ samples.
\end{Lemma}
Note that even though the above algorithm does not use distribution $q$, it finds $i$ such that $p(i) - q(i) \geq \beta p(i)$
just by the properties of $\ell_1$ distance.
Furthermore, $\beta_j$ increases with $j$ and $\alpha_j$ decreases with $j$; thus the above lemma states that the algorithm finds an element $i$  such that either $(p(i) -q(i))/p(i)$ is large, but may not be heavy, or $(p(i) -q(i))/p(i)$ is small, yet it belongs to one of the higher probabilities. This precise trade-off becomes important to bound the sample complexity.
\subsection{Testing for near-uniform distributions}
We define a distribution $p$ to be \emph{near-uniform} if
$\max_i p(i) \leq 2 \min_i p(i)$. Recall that we need to find a distinguishing element and a distinguishing set. As we show, for near-uniform distributions, there are singleton distinguishing sets and hence are easy to find.
Using \textsc{Find-element}, we first define a meta algorithm
to test for near-uniform distributions. The inputs to the algorithm
are parameter $\epsilon$, error $\delta$, distributions $p,q$ and an element $y$ such that $p(y) \geq q(y)$.  Since we use \textsc{Near-uniform-identity-test} as a subroutine later,  
$y$ is given from the main algorithm. However, if we  want to use \textsc{Near-uniform-identity-test} by itself,
 we can find a $y$ using \textsc{Find-element}$(\epsilon,p)$.

The algorithm uses \textsc{Find-element} to find an element $x$ such that
$q(x) - p(x) \geq \beta q(x)$. Since $p(y) \geq q(y)$ and 
$q(x) - p(x) \geq \beta q(x)$, running \textsc{Test-equal} on set 
$\{x,y\}$ will yield an algorithm for identity testing.
The precise bounds in Lemmas~\ref{lem:test_equal} and~\ref{lem:find_far} help  us to obtain the optimal sample complexity.
In particular,
\begin{Lemma}[Appendix~\ref{app:uniformidentitytesting}]
\label{lem:uniformidentitytesting}
If $p=q$, then \textsc{Near-uniform-identity-test} returns \texttt{same} with probability $\geq 1-\delta$.
If $p$ is near-uniform and $\norm{p-q}_1 \geq \epsilon$, 
then
\textsc{Near-uniform-identity-test} returns \texttt{diff}
with probability $\geq 1/5-\delta$. The algorithm uses $\cO \left(\frac{1}{\epsilon^2} \cdot\log \frac{1}{\delta\epsilon}\right)$ samples.
\end{Lemma}

\begin{center}
\fbox{\begin{minipage}{1.0\textwidth}
Algorithm \textsc{Near-uniform-identity-test} \newline
\textbf{Input:} distance $\epsilon$, error $\delta$, distributions $p,q$, 
an element $y$ such that $p(y)\ge q(y)$.
\begin{enumerate}
\item 
Run \textsc{Find-element}$(\epsilon,q)$ to obtain tuples $(x_j,\beta_j,\alpha_j)$ for $1 \leq j \leq 16/\epsilon$.
\item
For every tuple $(x_j,\beta_j,\alpha_j)$, run \textsc{Test-equal}$(\beta^2_j/144, 6\delta/(\pi^2j^2), p_{\{x,y\}},q_{\{x,y\}})$.
\item Output \texttt{same} if \textsc{Test-equal}
in previous step returns \texttt{same} for all tuples,
 otherwise output \texttt{diff}.
\end{enumerate}
\end{minipage}}
\end{center} 

\subsection{Finding a distinguishing set for general distributions}
We now extend \textsc{Near-uniform-identity-test} to general distributions.
Recall that we need to find a distinguishing element and a distinguishing set.

Once we have an element such that $p(i) > q(i)$, our  objective is  to find a 
distinguishing set $S$ such that $p(S) < q(S)$ and $p(S) \approx p(i)$. 
Natural candidates for such sets are combinations of elements whose probabilities $\leq p(i)$. Since $p$ is known, we can select such sets easily. 
Let $G_i = \{j: j \geq i\}$.
Consider the sets $H_1,H_2,\ldots$ formed by combining 
elements in $G_i$ such that $p(i) \leq p(H_j) \leq 2p(i), \forall j$. 
We ideally would like to use one of these $H_j$s as $S$, however
depending on the values of $p(H_j)$ 
three possible scenarios arise and that constitutes the main algorithm.

We need one more definition for describing the main identity test. For any distribution $p$,
 and a partition of $S$ into disjoint subsets $\cS = \{S_1,S_2,..\}$, the induced distribution $p^{\cS}$ 
is a distribution over $S_1, S_2,\ldots $ such that $\forall i, \, p^{\cS}_{S}(S_i) = \frac{p(S_i)}{p(S)}$.
\ignore{Note that induced distributions follow triangle inequality, \ie 
for any partition $\cS$ over $S$, $\norm{p^{\cS}_S - q^{\cS}_S}_1 \leq
\norm{p_S - q_S}_1$.}

\subsection{Proposed identity test}
The algorithm is a combination of tests for each possible scenarios.
First it finds a set of tuples $(i,\beta,\alpha)$ such that one tuple satisfies $(p(i)-q(i))/p(i) \geq \beta$ and $i$ is $\alpha$-heavy.
Then, it divides $G_i$ into $H_1,H_2,\ldots$ such that $, \forall j, \,p(i) \leq p(H_j) \leq 2p(i)$. 
If $\norm{p-q}_1 \geq \epsilon$, then there are three possible cases. 
\begin{enumerate}
\item
$p(H_j)(1-\beta/2) \leq q(H_j)$ for most $j$s.
We can randomly pick a set $H_j$ 
and sample from $H_j \cup \{i\}$ and  we would be able to test if $\norm{p-q}_1 \geq \epsilon$
using $n(i)$ when sampled from $H_j \cup\{i\}$.
\item
$p(H_j)(1-\beta/2) \geq q(H_j)$ for most $j$.
Since for most $j$'s, $p(H_j)(1-\beta/2) \geq q(H_j)$,
we have $p(G_i) (1-\beta/2) \geq q(G_i)$, and
since $p(G_i)\geq \alpha$, we can sample from the entire distribution
and use $n(G_i)$ to test if $\norm{p-q}_1 \geq \epsilon$.
\item 
For some $j$, $p(H_j)(1-\beta/2) \geq q(H_j)$ and 
for some $j$, $p(H_j)(1-\beta/2) \leq q(H_j)$.
It can be shown that 
this condition implies that 
elements in $G_i$ can be 
grouped into $H_1,H_2,\ldots$ such that
induced distribution on groups is near-uniform and yet the 
$\ell_1$ distance between the induced distributions is large. We use \textsc{Near-uniform-identity-test} 
for this scenario.
\end{enumerate}
The algorithm has a step corresponding to each of the above three scenarios.
If $p=q$, then all three steps would output \texttt{same} with high probability, otherwise one of the steps would output~\texttt{diff}.
The main result of this section is to bound the sample complexity of \textsc{Identity-test}
\begin{Theorem}[Appendix~\ref{app:thmidentity}]
\label{thm:identity}
If $p=q$, then \textsc{Identity-test} returns \texttt{same} with probability $\geq 1-\delta$
and if $\norm{p-q}_1 \geq \epsilon$, then \textsc{Identity-test} returns \texttt{diff}
with probability $\geq 1/30$.
The algorithm uses at most 
$\sccike \leq \Theta\left( \frac{1}{\epsilon^2} \cdot \log^2 \frac{1}{\epsilon} \cdot  \log \frac{1}{\epsilon\delta} \right)$ samples.
\end{Theorem}
The proposed identity testing has different error probabilities when $p=q$ and $\norm{p-q}_1 \geq \epsilon$. In particular, if $p=q$, the algorithm returns \texttt{same} with probability $\geq 1-\delta$ and if $\norm{p-q}_1 \geq \epsilon$ it outputs
\texttt{diff} with probability $\geq 1/30$. While the probability of success
for $\norm{p-q}_1\geq \epsilon$ is small, it
can be boosted arbitrarily close to $1$, by repeating
the algorithm $\cO(\log (1/\delta))$ times
and testing if more than $1/60$ fraction of times the algorithm outputs \texttt{diff}.
By a simple Chernoff type argument, it can be shown that 
for both cases $p=q$ and $\norm{p-q}_1$, the error probability of 
the boosted algorithm is $\leq \delta$.
Furthermore, throughout the paper we have calculated all the constants except sample complexities which we have left in $\cO$ notation.

\begin{center}
\fbox{\begin{minipage}{1.0\textwidth}
Algorithm \textsc{Identity-test} \newline
\textbf{Input:} error $\delta$, distance $\epsilon$ an unknown distribution $q$,  and a known distribution $p$.
\begin{enumerate}
\item Run \textsc{Find-element} $(\epsilon,p)$ to obtain tuples $(x,\beta,\alpha)$.
\item For every tuple $(x,\beta,\alpha)$:
\begin{enumerate}
\item Let $G_x=\{y: y\geq x\}$.
\item Partition $G_x$ into groups $\cH= H_1,H_2,\ldots $s.t. for each group $H_j$, $p(x) \leq p(H_j) \leq 2p(x)$.
\item Take a random sample $y$ from $p^{\cH}_{G_x}$ and run \textsc{Test-equal}$\bigl(\frac{\beta^2}{1800},\frac{\epsilon\delta}{48},p_{\{x,y\}},q_{\{x,y\}}\bigr)$.
\item Run \textsc{Test-equal}$\bigl(\bigl(\frac{\alpha\beta}{5}\bigr)^2,\frac{\epsilon\delta}{48},p^{\{G_x,G^c_x\}},q^{\{G_x,G^c_x\}}\bigr)$.
\item Run \textsc{Near-uniform-identity-test}$\bigl(\frac{\beta}{5},\frac{\epsilon\delta}{48},p^{\cH}_{G_x},q^{\cH}_{G_x}\bigr)$.
\end{enumerate}
\item Output \texttt{diff} if any of the above tests returns \texttt{diff} for any tuple, otherwise output \texttt{same}.
\end{enumerate}
\end{minipage}}
\end{center} 

\section{Closeness testing}
\label{sec:closeness}
Recall that in closeness testing, both $p$ and $q$ are unknown and we 
test if $p=q$ or $\norm{p-q}_1 \geq \epsilon$ using samples.
First we relate identity testing to closeness testing.

Identity testing had two parts: finding a distinguishing element $i$ and a distinguishing set $S$.
The algorithm we used to generate $i$ did not use any a priori knowledge of the distribution.
Hence it carries over to closeness testing easily. The main difficulty of extending
identity testing to closeness testing is to find a distinguishing set.
Recall that in identity testing, we ordered elements 
such that their probabilities are decreasing and considered set $G_i = \{j: j \geq i\}$
to find a distinguishing set. 
$G_i$ was known in identity testing, however in closeness testing, it is unknown
and is difficult to find. 

The rest of the section is organized as follows:
We first outline a method of identifying a distinguishing set by sampling at a certain frequency (which is unknown).
We then formalize finding a distinguishing element and then
show how one can use a binary search to find the sampling frequency and a distinguishing set. We finally describe our main closeness test, which requires few additional techniques to handle some \emph{special cases}.

\subsection{Outline for finding a distinguishing set}

Recall that in identity testing, we ordered elements 
such that their probabilities are decreasing  and considered $G_i = \{j: j \geq i\}$. We then used a subset of $S \subset G_i$ such that $p(S) \approx p(i)$
as the distinguishing set. However, in closeness test this is not possible as set $G_i$ is unknown. We now outline a method of finding such a set $S$ using random sampling without the knowledge of $G_i$.

Without loss of generality, assume that elements are ordered such that 
$p(1)+q(1) \geq p(2) + q(2) \geq \ldots \geq p(k)+q(k)$. The algorithm does 
not use this fact and the assumption is for the ease of proof notation.
Let $G_i = \{j : j \geq i\}$ under this ordering ($G_i$ serves same purpose as $G_i$ for identity testing, 
however is symmetric with respect to $p,q$ and hence easy to handle compared to that of identity testing). 
Furthermore, for simplicity in the rest of the section, assume that $p(i) > q(i)$ and $p(G_i) \leq q(G_i)$. 
Suppose we come up with a scheme that finds subset $S$ of $G_i$
such that $p(S) \approx p(i)$ and $p(S) < q(S)$,
then as in \textsc{Identity-test}, we can use that scheme together with \textsc{Test-equal} on $S\cup \{i\}$ to differentiate between $p=q$ and $\norm{p-q}_1 \geq \epsilon$.

The main challenge of the algorithm is to find a distinguishing subset of $G_i$. Let $r= (p+q)/2$, \ie $r(j) = (p(j)+q(j))/2 \, \, \forall \, 1\leq j \leq k$.
Suppose we know $r_{0} = \frac{r(i)}{r(G_i)}$. 
Consider a set $S$ formed  by including each element $j$ independently with probability $r_0$. Thus the probability of that set can be written as
\[
p(S) = \sum^k_{j=1} \II_{j \in S} p(j),
\]
where $\II_{j \in S}$ is the indicator random variable for $j \in S$.
In any such set $S$, there might be elements that are not from $G_i$.
We can prune these elements (refer to them as $j'$) by sampling from the distribution 
$p_{\{j,j'\}}$ and testing if $j'$ appeared more than $j$.  Precise probabilistic arguments are given later.
Suppose we remove all elements in $S$ that are not in $G_i$. Then,
\[
p(S) =  \sum_{j \in G_i} \II_{j \in S} p(j).
\]
Since $\Pr( \II_{j \in S} =1) = r_0$,
\begin{equation*}
\label{eq:expectation}
\EE[p(S)] = \sum_{j \in G_i} \EE[\II_{j \in S}] p(j) = r_0 \sum_{j \in G_i} p(j) = \frac{r(i)}{r(G_i)} \cdot p(G_i).
\end{equation*}
Similarly one can show that $\EE[q(S)] = \frac{ r(i)}{r(G_i)} \cdot q(G_i)$. Thus $\EE[p(S)] < \EE[q(S)]$ and 
$\EE[p(S)]+\EE[q(S)] = p(i)+q(i)$. 
Note that for efficiently using \textsc{Test-equal}, we 
not only need $p(i) > q(i)$ and  $\EE[p(S)] < \EE[q(S)]$, 
but we the chi-squared distance needs to be large. 
It can be shown that this condition is same as stating $p(S)+q(S) \approx p(i) + q(i)$ is necessary and hence $\EE[p(S)]+\EE[q(S)] = p(i)+q(i)$ is useful.

Thus in expectation $S$ is a good candidate for distinguishing set.
Hence if we take samples from $S \cup \{i\}$ 
and compare $p(S), p(i)$ and $q(S), q(i)$, we can test if $p=q$
or $\norm{p-q}_1 \geq \epsilon$.

We therefore have to find  an $i$ such that $p(i) > q(i)$ and $p(G_i)  < q(G_i)$,
estimate $r(i)/r(G_i)$ and convert the above expectation argument to a probabilistic one.
While the calculations and analysis in expectation seem natural, judiciously analyzing 
the success probability of these events takes a fair amount of effort. Furthermore, note that given a conditional sampling access to $p$ and $q$, one can generate a conditional sample from $r$, by selecting $p$ or $q$ independently with probability $1/2$ and then obtaining a conditional sample from the selected distribution.
\subsection{Finding a distinguishing element $i$}
We now show that using an algorithm similar to \textsc{Find-element}, we can find 
an $i$ such that ($p(i) > q(i)$ and $p(G_i)  \leq q(G_i)$) or ($p(i) < q(i)$ and $p(G_i) > q(G_i)$). 
To quantify the above statement we 
\ignore{We now outline each step of the algorithm. The first step of the algorithm
finds element $i$ such that $p(G_i) \leq q(G_i)$ and $p(i)> q(i)$. \ignore{Furthermore for symmetry, we define $r$ to be 
the average of $p$ and $q$, \ie $r(i) = \frac{p(i)+q(i)}{2}, \forall i$. Furthermore, redefine  $G_i = \{j: r(j) \leq  r(i)\}$. }}
need the following definition. We define $\beta$-approximability as
\begin{Definition}
For a pair of distributions $p$ and $q$, 
element $i$ is $\beta$-approximable,
if
\[
\left \lvert \frac{p(i) - q(i)}{p(i)+q(i)} - \frac{p(G_i) - q(G_i)}{p(G_i) +q(G_i)} \right  \rvert \geq \beta.
\]
\end{Definition}
As we show later, it is sufficient to consider $\beta$-approximable elements instead of elements with  $p(i) > q(i)$ and $p(G_i)  \leq q(G_i)$.
Thus the first step of our algorithm is to find $\beta$-approximable elements. To this end, we show that 
\begin{Lemma}[Appendix~\ref{app:exp_approx}]
\label{lem:exp_approx}
If $\norm{p-q}_1 \geq \epsilon$, then
\[
\sum_{i} \frac{p(i)+q(i)}{2} \left \lvert \frac{p(i) - q(i)}{p(i)+q(i)} - \frac{p(G_i) - q(G_i)}{p(G_i) +q(G_i)} \right  \rvert \geq \frac{\epsilon}{4}.
\]
\end{Lemma}
Hence if we use \textsc{Find-element} for the distribution $r = (p+q)/2$, then one of the tuples would be $\beta_j$-approximable for some $\beta_j$. 
Note that with $a_i = \left \lvert \frac{p(i) - q(i)}{p(i)+q(i)} - \frac{p(G_i) - q(G_i)}{p(G_i) +q(G_i)} \right  \rvert$, $0 \leq a_i \leq 2$ and $\sum^k_{i=1} r(i)a_i \geq \epsilon/4$.
By Lemma~\ref{lem:find_far}, \textsc{Find-element} outputs a tuple $(i,\alpha,\beta)$ such that $i$ that is $\alpha$-heavy
and $\beta$-approximable. Note that although we obtain $i$ and guarantees on $G_i$, the algorithm does not find $G_i$.
\begin{Lemma}
\label{lem:find_approx}
With probability $\geq 1/5$, of the tuples returned by~\textsc{Find-element}$(\epsilon,r)$ there exist at least one tuple that is both $\alpha$-heavy and $\beta$-approximable. 
\end{Lemma}

\subsection{Approximating $\frac{r(i)}{r(G_i)}$ via binary search}
Our next goal is to estimate $r_0 = \frac{r(i)}{r(G_i)}$ using samples.
It can be easily shown that it is sufficient to  know $\frac{r(i)}{r(G_i)}$ up-to a multiplicative factor, say $\gamma$ (we later choose  $\gamma=\Theta(\log \log \log k)$).
Furthermore by the definition of $G_i$,  $r(G_i) \geq r(i)$
and $r(G_i) = \sum_{j \geq i} r(j) \leq  \sum_{j \geq i } r(i) \leq k r(i)$. Therefore,
\[
 \frac{1}{k} \leq \frac{r(i)}{r(G_i)} \leq 1,
\]
and $\log k \geq - \log \frac{r(i)}{r(G_i)} \geq 0$. Approximating $\frac{r(i)}{r(G_i)}$
up-to a multiplicative factor $\gamma$ is the same as approximating  $ \log \frac{r(i)}{r(G_i)} $ up-to an additive factor $\log \gamma$. We can thus run our algorithm for $ \frac{r(i)}{r(G_i)}$ corresponding to each value of $\{0, \log \gamma, 2 \log \gamma, 3 \log \gamma,\ldots, \log k\}$,
and if $\norm{p-q}_1 \geq \epsilon$, at least for one value of $\frac{r(i)}{r(G_i)}$ we output \texttt{diff}. Using carefully chosen thresholds we can also ensure that if $p=q$, the algorithm outputs \texttt{same} always. The sample complexity for the above algorithm is $\frac{\log k}{\log \gamma} \approx \tilde \Theta(\log k)$ times the complexity when we know $\frac{r(i)}{r(G_i)}$.
We improve the sample complexity by using a better search algorithm over $\{0, \log \gamma, 2 \log \gamma, 3 \log \gamma,\ldots, \log k\}$. We develop a comparator (step $4$ in \textsc{Binary-search}) with the following property: if our guess value $r_{\guess} \geq \gamma\frac{r(i)}{r(G_i)}$ it outputs \texttt{heavy} and if $r_{\guess} \leq \frac{1}{\gamma} \cdot \frac{r(i)}{r(G_i)}$ it outputs \texttt{light}.
 Using such a comparator, we do a binary search and find the right value faster. Recall that binary search over $m$ elements uses $\log m$ queries. For our problem $m = \log k$ and thus our sample complexity is approximately  $\log \log k$ times the sample complexity of the case when we know $\frac{r(i)}{r(G_i)}$. \ignore{In Appendix~\ref{app:closeness}, we show that this sample complexity is $\tcO(\log \log k)$.}

However, our comparator cannot identify if we have a good guess \ie if  $\frac{1}{\gamma} r_{\guess} \frac{r(i)}{r(G_i)} \leq r_{\guess} \leq \gamma \cdot \frac{r(i)}{r(G_i)}$. Thus, our binary search instead of outputting the value of $\frac{r(i)}{r(G_i)}$ up-to some approximation factor $\gamma$, 
finds a set of candidates $r^1_{\guess}, r^2_{\guess},\ldots$ such that at least one of the $r^j_\guess$s satisfies
\[
\frac{1}{\gamma}\frac{r(i)}{r(G_i)}
\leq
r^j_{\guess}
\leq 
\gamma \frac{r(i)}{r(G_i)}.
\]
Hence, for each value of $r_{\guess}$ we assume that $r_{\guess} \approx r(i)/r(G_i)$, and run the closeness test. At least for one value of $r_{\guess}$ we would be correct. The algorithm is given in \textsc{Binary-search}. 

The algorithm \textsc{Prune-set} removes all elements of probability $\geq 4r(i)$, yet does not remove any element of probability $\leq r(i)$. Since after pruning  $S$ only contains elements of probability $\leq 4r(i)$, we show that at some point of the $\log \log k$ steps, the algorithm encounters $r_{\guess} \approx \frac{r(i)}{r(G_i)}$.

\ignore{
\begin{center}
\fbox{\begin{minipage}{1.0\textwidth}
Algorithm \textsc{Closeness-test} \newline
\textbf{Input:} $\epsilon$, oracles $p,q$.
\textbf{parameters:} 
\begin{enumerate}
\item 
Generate a set of tuples using \textsc{Find-elements} using the distribution $r$.
\item 
For every tuple, run \textsc{Assisted-closeness-test}.
\item
If any of them returned \texttt{diff}, output \texttt{diff},
else \texttt{same}.
\end{enumerate}
\end{minipage}}
\end{center}
The main sub-routine for the algorithm is the binary search for estimating the value of 
$\frac{r(i)}{r(G_i)}$. At each step the algorithm finds sets $S_1,S_2\ldots S_m$.
It first prunes the sets to remove all heavy elements using $\textsc{Prune-set}$.
If $r_{\guess}$ is close to the underlying $\frac{r(i)}{r(G_i)}$,
then we show that for at least one of the underlying sets $S$
has a subset which we can use to differentiate between the two hypothesis.
The algorithm also has access to a comparator \textsc{Test-heavy}, that outputs 
\texttt{heavy} if $r_{\guess}$ is much higher than the actual value 
and outputs \texttt{light} if $r_{\guess}$ is much smaller than the actual value.
Thus the algorithm conducts a binary search over all possible values of $\frac{r(i)}{r(G_i)}$,
resulting a sample complexity of $\log \log k$.
}

\begin{center}
\fbox{\begin{minipage}{1.0\textwidth}
Algorithm \textsc{Prune-set} \newline
\textbf{Input:} $S$, $\epsilon$, $i$, $\alpha$, $m$, and $\gamma$.
\newline
\textbf{Parameters:} $\delta' = \frac{\delta}{40 m \log \log k}$, $n_1 = \cO \left(\left(\log \frac{\gamma}{\delta'\alpha\beta} \right)\cdot 
 \left( \frac{\gamma}{\alpha\beta} \log \frac{\gamma}{\alpha \beta}+\log\frac{1}{\delta'}\log\log \frac{1}{\delta'}\right) \right)$, $n_2 = \cO(\log \log \log k + \log \frac{1}{\epsilon\delta})$.
 \newline
 Repeat $n_1$ times:
\newline
Obtain a sample $j$ from $r_S$ and 
sample $n_2$ times from $r_{\{j,i\}}$. If $n(j) \geq 3n_2/4$, remove $j$ from set $S$.
\end{minipage}}
\end{center} 
\begin{center}
\fbox{\begin{minipage}{1.0\textwidth}
Algorithm \textsc{Binary-search} \newline
\textbf{Input:} Tuple $(i,\beta,\alpha)$.
\newline
\textbf{Parameters:} $\gamma = 1000 \log \frac{ \log \log k}{\delta\epsilon}$, $n_3 = \cO \left(\gamma^2 \log \frac{\log \log k}{\delta} \right) $.
\newline Initialize $\log r_{\guess} = -\log \sqrt{k}$.
Set $\lowvalue = -\log k$ and $\highvalue = 0$. Do $\log \log k$ times:
\begin{enumerate}
\item Create a set $S$ 
by independently keeping elements $\{1,2,\ldots,k\} \setminus \{i\}$ each w.p. $r_{\guess}$.
\item Prune $S$ using \textsc{Prune-set}$(S,\epsilon,i,\alpha,1,\gamma)$.
\item 
Run \textsc{Assisted-closeness-test}$(r_{\guess}, (i,\beta,\alpha), \gamma, \epsilon,\delta)$.
\item 
Obtain $n_3$ samples from $S\cup \{i\}$.
If $n(i) < \frac{5n_3}{\gamma}$, then output \texttt{heavy}, else output \texttt{light}.
\begin{enumerate}
\item If output is \texttt{heavy}, update $\high = \log r_{\guess}$ and 
$\log r_{\guess} = (\log r_{\guess}+ \low)/2$.
\item If output is \texttt{light}, update
$\low = \log r_{\guess}$ and
$\log r_{\guess} = (\log r_{\guess}+ \high)/2$.
\end{enumerate}
\item 
If any of the \textsc{Assisted-closeness-test}s return \texttt{diff}, then output \texttt{diff}.
\end{enumerate}
\end{minipage}}
\end{center}

\ignore{
\begin{center}
\fbox{\begin{minipage}{1.0\textwidth}
Algorithm \textsc{Test-heavy} \newline
\textbf{Input:} set $S$, element $i$, $\alpha$.
\newline
\textbf{Parameters:} $n= \cO(\frac{1}{\gamma^2} \log \log k)$.
\begin{enumerate}
\item
Obtain $n$ samples from the entire set $S$.
If $n(i) < \frac{4n}{\gamma}$, then output \texttt{heavy}, else output \texttt{light}.
\end{enumerate}
\end{minipage}}
\end{center} 
}

\begin{Lemma}[Appendix~\ref{app:reaches}]
\label{lem:reaches}
If $i$ is $\alpha$-heavy and $\beta$-approximable,
 then the algorithm~\textsc{Binary-search},
with probability $\geq 1 - \delta$, 
\ignore{During the run of the algorithm 
if~\textsc{Test-heavy} always outputs 
\texttt{heavy} whenever $r_{\guess} >\gamma \frac{r(i)}{r(G_i)}$ 
and \texttt{light} whenever $r_{\guess} < \gamma^{-1} \frac{r(i)}{r(G_i)}$,}
reaches $r_{\guess}$ such that 
\[
 \frac{r(i)}{\gamma} = \frac{r(G_i)}{\gamma} \cdot \frac{r(i)}{r(G_i)}\leq r_{\guess} \leq 
\frac{\gamma}{\beta} \cdot \frac{r(i)}{r(G_i)}.\]
\end{Lemma}
Note that due to technical reasons we get an additional $1/\beta$
factor in the upper bound and a factor of  $r(G_i)$ in the lower bound.
\subsection{Assisted closeness test}
We now discuss the proposed test, which uses the above value of $r_{\guess}$. As stated before, in expectation
it would be sufficient to keep elements in the set $S$ with probability $r_{\guess}$ and use the resulting set $S$ to test for closeness.
However, there are two caveats. Firstly, \textsc{Prune-set} can remove only elements
which are bigger than $4(i)$, while we can reduce the factor $4$ to any number $>1$,
but we can never reduce it to $1$ as if there is an element with probability $1+\delta'$ for sufficiently small $\delta'$,
that element is almost indistinguishable from an element with probability $1-\delta'$.
Thus we need a way of ensuring that elements with probability $> r(i)$ and $\leq 4r(i)$  do not affect the concentration inequalities.

Secondly, since we have an approximate value of $r(i)/r(G_i)$, the probability that  required quantities concentrate is small and 
we have to repeat it many times to obtain a higher probability of success.
Our algorithm address both these issues and is given below:

The algorithm picks $m$ sets and prunes them to ensure that
none of the elements has probability $\geq 4r(i)$ and considers two possibilities: there exist many elements $j$ such that $j \notin G_i$ and
\[
\lV \frac{p(i) - q(i)}{r(i)} - \frac{p(j) - q(j)}{r(j)} \rV \geq \beta''\,\, \text{($\beta''$ determined later)},
\]
or the number of such elements is small.
If it is the first case, the algorithm finds such an element $j$ and performs \textsc{Test-equal} over set $\{i,j\}$. Otherwise, we show that $r(S) \approx r(i)$, it concentrates, and with high probability
\[
\lV \frac{p(i) - q(i)}{r(i)} - \frac{p(S) - q(S)}{r(S)} \rV \geq \beta'' \,\, \text{($\beta''$ determined later)},
\]
and thus one can sample from $S \cup \{i\}$ and use $n(i)$ to test closeness.

To conclude, the proposed~\textsc{Closeness-test} uses~\textsc{Find-element} to find a distinguishing element $i$.
It then runs \textsc{Binary-search} to approximate $r(i)/r(G_i)$. However since the search does not identify
if it has found a good estimate of $r(i)/r(G_i)$, for each estimate it runs ~\textsc{Assisted-closeness-test}
which uses the distinguishing element $i$ and the estimate of $r(i)/r(G_i)$.
The main result in this section is the sample complexity of our proposed~\textsc{Closeness-test}.

\begin{Theorem}[Appendix~\ref{app:thmcloseness}]
\label{thm:closeness}
If $p=q$, then \textsc{Closeness-test} returns \texttt{same} with probability $\geq 1-\delta$
and if $\norm{p-q}_1 \geq \epsilon$, then \textsc{Closeness-test} returns \texttt{diff}
with probability $\geq 1/30$.
The algorithm uses 
$\scccke \leq \tcO \left( \frac{\log \log k}{\epsilon^5} \right)$ samples.
\end{Theorem}
As stated in the previous section, by repeating and taking a majority, the success probability can be boosted arbitrarily close to $1$.
Note that none of the constants or the error probabilities have been optimized. Constants for all the parameters except the sample complexities $n_1,n_2,n_3,$ and $n_4$ have been given.

\begin{center}
\fbox{\begin{minipage}{1.0\textwidth}
Algorithm \textsc{Closeness-test} \newline
\textbf{Input:} $\epsilon$, oracles $p,q$.
\begin{enumerate}
\item 
Generate a set of tuples using \textsc{Find-element}$(\epsilon,r)$.
\item 
For every tuple $(i,\alpha,\beta)$, run \textsc{Binary-search}$(i,\beta,\alpha)$.
\item
If any of the \textsc{Binary-search} returned \texttt{diff} output \texttt{diff} otherwise output \texttt{same}.
\end{enumerate}
\end{minipage}}
\end{center}
\begin{center}
\fbox{\begin{minipage}{1.0\textwidth}
Algorithm \textsc{Assisted-closeness-test} \newline
\textbf{Input:}  $r_{\guess}$, tuple $(i,\beta,\alpha)$, $\gamma$, $\epsilon$, and $\delta$.
\newline
\textbf{Parameters:} 
$\beta'' = \frac{\alpha\beta}{128\gamma \log \frac{128\gamma}{\beta^2}}$, $m = \frac{4096\gamma}{\alpha \beta^2}$, $n_4 = O( \gamma/( \alpha \beta))$, 
and $\delta' =\frac{\epsilon\delta}{32m(n_4+1) \log \log k} $.
\begin{enumerate}
\item Create $S_1, S_2,\ldots S_m$ independently by keeping elements $\{1,2,\ldots,k\} \setminus \{i\}$ each w.p. $r_{\guess}$.
\item Run \textsc{Prune-set}$(S_\ell,\epsilon,i,\alpha,m,\gamma)$ for $1\leq \ell \leq m$
\item For each set $S$ do:
\begin{enumerate}
\item Take $n_4$ samples from $r_{S\cup\{i\}}$ 
and 
for all  seen elements $j$, run \textsc{Test-equal} $\left( (\beta'')^2/25 , \delta', p_{\{i,j\}}, q_{\{i,j\}}\right)$.
\item
Let $\cS=\{ \{i\},S\}$. Run \textsc{Test-equal}$\left( \frac{\alpha^3\beta^3}{2^{23} \gamma^2\log^3 \frac{128}{\gamma\beta^2}}  , \delta', p^{\cS}_{S_1\cup \{i\}}, q^{\cS}_{S_1\cup \{i\}}\right)$.
\end{enumerate}
\item 
If any of the above tests return \texttt{diff}, output 
\texttt{diff}.
\end{enumerate}
\end{minipage}}
\end{center}

\section{Acknowledgements}
\label{sec:ack}

We thank Jayadev Acharya, Cl\'ement Canonne, Sreechakra Goparaju, and
Himanshu Tyagi for useful suggestions and discussions.

\bibliographystyle{abbrvnat}
\bibliography{abr,masterref}
\appendix
\section{Tools}
We use the following variation of the Chernoff bound.
\begin{Lemma}[Chernoff bound]
\label{lem:eqchern}
If $X_1, X_2 \ldots X_n$ are distributed according to Bernoulli $p$, then 
\begin {align*}
\Pr\left(\frac{\sum_{i=1}^n X_i}{n}-p>\delta\right)&  \leq e^{-2n\delta^2},\\
\Pr \left(\frac{\sum_{i=1}^n X_i}{n}-p<-\delta\right) &\leq e^{-2n\delta^2}.
\end{align*}
\end{Lemma}

The following lemma follows from a bound in~\cite{AcharyaDJOPS12} and the fact that $y^5 e^{-y}$ is bounded for all non-negative values of $y$.
\begin{Lemma}[\cite{AcharyaDJOPS12}]
\label{lem:var}
  For two independent Poisson random variables $\mu$ and $\mu'$ with means $\lambda$ and $\lambda'$ respectively,
\begin{align*}
  E \left( \frac{(\mu - \mu')^2 - \mu - \mu'}{\mu + \mu' -1} \right) & = \frac{(\lambda-\lambda')^2 }{\lambda+\lambda'} \left( 1 - e^{-\lambda-\lambda'}  \right),\\
\Var \left( \frac{(\mu - \mu')^2 - \mu - \mu'}{\mu + \mu' -1} \right) & \leq  4 \frac{(\lambda-\lambda')^2}{\lambda+\lambda'} + 
c^2,
\end{align*}
where $c$ is a universal constant.
\end{Lemma}

\section{Identity testing proofs}
\label{app:identity}
\subsection{Proof of Lemma~\ref{lem:test_equal}}
\label{app:test_equal}
Let 
$t = \frac{(n_1 -n_2)^2 - n_1-n_2}{n_1+n_2-1} 
+
\frac{(n_1 -n_2)^2 - n_1-n_2}{n'+n'' - n_1-n_2-1}$.
Since we are using $\poi(n)$ samples, $n_1,n_2,n'-n_1,n''-n_2$ are all independent Poisson distributions with means $np,nq,n(1-p)$, 
and $n(1-q)$ respectively.
Suppose the underlying hypothesis is $p=q$. 
By Lemma~\ref{lem:var}, $\EE[t]=0$
and since , $n_1,n_2,n'-n_1,n''-n_2$ are all independent Poisson distributions, variance of $t$ is the sum of variances of each term and hence $\Var(t) \leq  2c^2$ for some universal constant $c$.
Thus by Chebyshev's inequality 
\[
\Pr(t \geq n\epsilon/2) \leq \frac{8c^2}{(n\epsilon)^2} \leq \frac{1}{3}.
\]
Hence by the Chernoff bound~\ref{lem:eqchern}, after $18 \log\frac{1}{\delta}$ repetitions
probability that the majority of outputs is \texttt{same} is $\geq 1-\delta$.
Suppose the underlying hypothesis is $\frac{(p-q)^2}{(p+q)(2-p-q)}>\epsilon$. Then by Lemma~\ref{lem:var}
\begin{align*}
\EE[t] 
&= \frac{n(p-q)^2}{p+q} \left(1-e^{-n(p+q)}  \right)
+ \frac{n(p-q)^2}{2-p-q} \left(1-e^{-n(2-p-q)}  \right)
\\
&\stackrel{(a)}{\geq} \frac{n(p-q)^2}{p+q} \left(1-e^{-n\epsilon}  \right)
+ \frac{n(p-q)^2}{2-p-q} \left(1-e^{-n\epsilon}  \right) \\
& {\geq}  \frac{2n(p-q)^2}{(p+q)(2-p-q)} \left(1-e^{-n\epsilon}  \right)\\
& \stackrel{(b)}{\geq} \frac{n(p-q)^2}{(p+q)(2-p-q)}.
\end{align*}
$(a)$ from the fact that $p+q \geq (p+q)\frac{(p-q)^2}{(p+q)^2} \geq \frac{(p-q)^2}{p+q} \geq \epsilon$ and similarly $2-p-q \geq \epsilon$. 
$(b)$ follows the fact that $n\epsilon \geq 10$.
Similarly the variance is 
\[
\Var(t) \leq c^2 + \frac{4n(p-q)^2}{p+q} + c^2 +  \frac{4n(p-q)^2}{2-p-q} =  2c^2 + \frac{8n(p-q)^2}{(p+q)(2-p-q)}.
\]
Thus again by Chebyshev's inequality,
\[
\Pr(t \leq n\epsilon/2) \leq \frac{ 2c^2 + \frac{8n(p-q)^2}{(p+q)(2-p-q)}}{( \frac{n(p-q)^2}{(p+q)(2-p-q)} - \frac{n\epsilon}{2})^2} \leq \frac{32c^2}{(n\epsilon)^2} + \frac{32}{n\epsilon} \leq \frac{1}{3}.
\]
The last inequality follows when $n \geq \frac{\max(192,20c)}{\epsilon}$. The lemma follows by the Chernoff bound argument as before.
\qedhere
\ignore{
Then $|p-q| \geq \frac{\epsilon}{2}$.
We want to have $\Pr\left(\lV\frac{\sum_{i=1}^{n} x_i}{n}-\frac{\sum_{i=1}^{n} y_i}{n}\rV < \frac{\epsilon}{4}\right) \leq \epsilon $. 
\begin{align}
\Pr\left(\lV\frac{\sum_{i=1}^{n} x_i}{n}-\frac{\sum_{i=1}^{n} y_i}{n} \rV < 
\frac{\epsilon}{4}\right) 
&= \Pr\left(\lV\frac{\sum_{i=1}^{n} x_i}{n}-p -\frac{\sum_{i=1}^{n} y_i}{n}+q+(p-q)\rV < \frac{\epsilon}{4}\right) \nonumber\\
& \leq \Pr\left(\lV\frac{\sum_{i=1}^{n} x_i}{n}-p -\frac{\sum_{i=1}^{n} y_i}{n}+q\rV > |p-q|-\frac{\epsilon}{4}\right)\nonumber \\
& \leq \Pr\left(\lV\frac{\sum_{i=1}^{n} x_i}{n}-p -\frac{\sum_{i=1}^{n} y_i}{n}+q\rV >\frac{\epsilon}{4}\right)\nonumber \\
& \leq \Pr\left(\lV\frac{\sum_{i=1}^{n} x_i}{n}-p\rV > \epsilon/8 \right)+ 
\Pr\left(\lV\frac{\sum_{i=1}^{n} y_i}{n}-q\rV > \epsilon/8 \right)\label{unionBound}\\
&\leq 2 e^{-2\epsilon^2 n/64},\label{chernoffBound}
\end{align}
where (\ref{unionBound}) is by union bound and (\ref{chernoffBound}) is by the Chernoff bound. 
Hence, if $n \geq \frac{32 \log \frac{2}{\epsilon}}{\epsilon^2}$,
then the above probability $ 2 e^{-\epsilon^2 n/32} \leq \epsilon$. 
By following a similar approach for the case $p=q$,
it can be shown that $n \geq \frac{32}{\epsilon^2} \ln{\frac{4}{\epsilon}}$
is sufficient to guarantee that $\Pr\left(\lV\frac{\sum_{i=1}^{n} x_i}{n}-\frac{\sum_{i=1}^{n} y_i}{n}\rV > \frac{\epsilon}{4}\right) \leq \epsilon $.
}

\subsection{Proof of Lemma~\ref{lem:find_far}}
\label{app:find_far}
Let $A_j$ be the event that $a_{x_j} \geq \beta_j$ and $x_j$ is $\alpha_j$-heavy. Since we choose each tuple $j$ independently at time $j$, events $A_j$s are independent. Hence,
\begin{align*}
\Pr(\cup A_j) & = 1 - \Pr(\cap A^c_j) \\
&  = 1 - \prod^m_{j=1} \Pr(A^c_j) \\
& = 1 - \prod^m_{j=1} (1-\Pr(A_j)) \\
& \geq 1- e^{-\sum^m_{j=1} \Pr(A_j)}.
\end{align*}
Let $B_j = \{i: a_i \geq \beta_j \}$.
Since all elements in $B_j$ count towards $A_j$ except the last $\alpha_j$ part, $\Pr(A_j) \geq p(B_j) - \alpha_j$.
Thus 
\begin{align*}
\sum^m_{j=1} \Pr(A_j) 
& \geq \sum^m_{j=1} p(B_j) - \sum^m_{j=1} \alpha_j \\
& \geq \sum^m_{j=1} p(B_j) - \frac{\log m}{4\log 16/\epsilon} \\
& \geq \sum^m_{j=1} p(B_j)  -  1/4.
\end{align*}
We now show that $\sum^m_{j=1} p(B_j) \geq 1/2$, thus proving that
$\sum^m_{j=1} \Pr(A_j) \geq 1/4$ and 
$\Pr(\cup_{j} A_j)  \geq 1/5$.
Since $\sum^k_{i=1} p(i) a_i \geq \epsilon/4$,
\begin{align*}
\sum^m_{i=1} p(i)a_i
& = 
\sum_{i: a_i \geq \epsilon/8} p(i)a_i
+ 
\sum_{i: a_i <\epsilon/8} p(i)a_i
\\
& \leq 
\sum_{i: a_i \geq\epsilon/8} p(i) a_i
+ 
\epsilon/8.
\end{align*}
Thus 
\[
\sum_{i: a_i \geq \epsilon/8} p(i) a_i \geq \epsilon/8,
\]
and
\[
\sum_{i : a_i \geq \epsilon/8} p(i) \left \lfloor\frac{8a_i}{\epsilon } \right \rfloor \geq 1/2.
\]
In $\sum^m_{j=1} p(B_j)$, each $p(i)$ is counted exactly 
$\left \lfloor\frac{8a_i}{\epsilon} \right \rfloor$ times,
thus 
\[
\sum^m_{i=1} p(B_j) \geq \frac{1}{2}. 
\]
\qedhere

\subsection{Proof of Lemma~\ref{lem:uniformidentitytesting}}
\label{app:uniformidentitytesting}
The proof uses the following auxiliary lemma.
Let $p_{i,j} \ed p_{\{i,j\}}(i) = \frac{p(i)}{p(i)+p(j)}$ denote the 
probability of $i$ under conditional sampling.
\begin{Lemma}
\label{lem:chilow}
If 
\begin{equation}
\label{eq:expand2}
\left \lvert \frac{p(i)-q(i)}{p(i)+q(i)} - \frac{p(j)-q(j)}{p(j)+q(j)} \right \rvert \geq \epsilon,
\end{equation}
then
\begin{align*}
& \frac{(p_{i,j} - q_{i,j})^2}{(p_{i,j} + q_{i,j})(2-p_{i,j} - q_{i,j})} \\
& \geq
\frac{\epsilon^2 (p(i)+q(i))^2 (p(j)+q(j))^2}
{4[p(i)(q(i)+q(j)) + q(i)(p(i)+p(j))][p(j)(q(i)+q(j)) + q(j)(p(i)+p(j))]} \\
& \geq \epsilon^2 \frac{(p(i)+q(i))(p(j)+q(j))}{4(p(i)+q(i)+p(j)+q(j))^2}.
\end{align*}
\end{Lemma}
\begin{proof}
Let $s(i) = p(i)+q(i)$ and $s(j) = p(j)+q(j)$.
Upon expanding,
\begin{equation}
\label{eq:expand}
\left \lvert \frac{p(i)-q(i)}{p(i)+q(i)} - \frac{p(j)-q(j)}{p(j)+q(j)} \right \rvert
=  
2\left \lvert \frac{p(i)q(j) - p(j) q(i)}{s(i)s(j)} \right \rvert.
\end{equation}
Furthermore, $p_{i,j} = \frac{p(i)}{p(i)+p(j)}$ and similarly $q_{i,j} = \frac{q(i)}{q(i)+q(j)}$. 
Hence, 
\begin{align*}
& \frac{(p_{i,j} - q_{i,j})^2}{(p_{i,j} + q_{i,j})(2-p_{i,j} - q_{i,j})}\\
& = 
\frac{(p(i)q(j)-q(i)p(j))^2}
{[p(i)(q(i)+q(j)) + q(i)(p(i)+p(j))][p(j)(q(i)+q(j)) + q(j)(p(i)+p(j))]} \\
& \stackrel{(a)}{\geq }
\frac{\epsilon^2 s^2(i) s^2(j)}
{4[p(i)(q(i)+q(j)) + q(i)(p(i)+p(j))][p(j)(q(i)+q(j)) + q(j)(p(i)+p(j))]} \\
& \stackrel{(b)}{\geq }
\frac{\epsilon^2 s^2(i) s^2(j)}{4(s(i)+s(j))^2 s(i) s(j)}\\
& =
\frac{\epsilon^2 s(i)s(j)}{4(s(i)+s(j))^2}.
\end{align*}
$(a)$ follows by Equations~\eqref{eq:expand2} and~\eqref{eq:expand}.
$(b)$ follows from $\max(p(i),q(i)) \leq s(i)$
and $\max(p(j),q(j)) \leq s(j)$.
\end{proof}

\begin{proof}(Lemma~\ref{lem:uniformidentitytesting})
We first show that if $p=q$, then \textsc{Near-uniform-identity-test} returns same with probability $\geq 1-\delta$. By Lemma~\ref{lem:test_equal}, \textsc{Test-equal} returns error probability $\delta_j=6\delta/(\pi^2j^2)$ for the $j$th tuple and hence by the union bound,  the overall error is $\leq \sum_j{\delta_j}\le \delta$.

If $p \neq q$, with probability $\geq 1/5$,  \textsc{Find-element} returns an element $x$ that is $\alpha$-heavy and $q(x) - p(x) \geq \beta q(x)$.
For this $x,y$, since  $p(y) \geq q(y)$,
\[
\frac{q(x)-p(x)}{p(x)+q(x)} - \frac{q(y)-p(y)}{p(y)+q(y)} 
 \geq \frac{\beta q(x)}{p(x)+q(x)}.
\]
By Lemma~\ref{lem:chilow} the chi-squared distance between $p_{\{x,y\}}$ and $q_{\{x,y\}}$ is lower bounded by
\begin{align*}
& \geq 
\left( \frac{\beta q(x)}{p(x)+q(x)} \right)^2 \frac{(p(x)+q(x))^2(p(y)+q(y))^2}
{4[p(x)(q(x)+q(y)) + q(x)(p(x)+p(y))][p(y)(q(x)+q(y)) + q(y)(p(x)+p(y))]} \\
& \stackrel{(a)}{\geq }
\frac{{\beta}^2 q^2(x)p^2(y)}
{4[p(x)(q(x)+p(y)) + q(x)(p(x)+p(y))][p(y)(q(x)+p(y)) + p(y)(p(x)+p(y))]} \\
& \stackrel{(b)}{\geq}
\frac{{\beta}^2 q^2(x)p^2(y)}
{4[2p(y)q(x)+p(x)p(y) + q(x)(2p(y)+p(y))][p(y)(q(x)+2p(x)) + p(y)(p(x)+2p(x))]} \\
& \stackrel{(c)}{\geq}
\frac{{\beta}^2 q^2(x)p^2(y)}
{4[2p(y)q(x)+q(x)p(y) + q(x)(2p(y)+p(y))][p(y)(q(x)+2q(x)) + p(y)(q(x)+2q(x))]} \\
& \stackrel{(d)}{\geq} \frac{{\beta}^2}{144}.
\end{align*}
$(a)$ follows from the fact that $p(y)\geq q(y)$. $p(x) \leq 2p(y)$ and $p(y) \leq 2p(x)$ hence $(b)$. $p(x) \leq q(x)$ implies $(c)$ and $(d)$ follows from 
numerical simplification.
\ignore{
\begin{align*}
\left( \frac{\beta' q(x)}{p(x)+q(x)} \right)^2\frac{(p(x)+q(x))(p(y)+q(y))}{4(p(x)+p(y)+q(x)+q(y))^2}
& \stackrel{(a)}{\geq }
{\beta'}^2 \frac{q^2(x) p(y)}{4(p(x)+q(x)){(p(x)+2p(y)+q(x))^2}} \\
& \stackrel{(b)}{\geq}
\frac{{\beta'}^2}{4(1-\beta')(2-\beta')} \frac{p(x)p(y)}{(p(x)(2-\beta') + (2-2\beta')p(y))^2} \\
& \stackrel{(c)}{\geq}
\frac{{\beta'}^2}{72(1-\beta')(2-\beta')} \geq \frac{\beta^2}{144},
\end{align*}
where $(a)$ follows from  $p(y) \geq q(y)$, $(b)$ follows from the fact that $p(x)(1-\beta') = q(x)$ and $(c)$ follows from 
the fact that $p$ is nearly-uniform, $p(x) \leq 2 p(y)$ and $p(y) \leq 2 p(x)$.
}
Thus by Lemma~\ref{lem:test_equal}
algorithm returns \texttt{diff} with probability $\ge 1-\delta$.
By the union bound, 
the total error probability is $\leq \frac45+\delta$.
The number of samples used is $16/\epsilon$ for the first step and 
 $\tcO \left(\frac{1}{\beta^2_j} \log \frac{1}{\delta_j} \right)$ for tuple $j$. Hence the total number of samples used is
\[
\frac{16}{\epsilon} + \sum^{16/\epsilon}_{j=1} \cO \left( \frac{1}{\beta^2_j} \log \frac{1}{\delta_j} \right) = \cO \left(\frac{1}{\epsilon^2}  \log \frac{1}{\delta\epsilon}\right).
\]
\end{proof}

\subsection{Proof of Theorem~\ref{thm:identity}}
\label{app:thmidentity}
We state the theorem statement for better readability:
If $p=q$, then \textsc{Identity-test} returns \texttt{same} with probability 
$\geq 1-\delta$ and if $\norm{p-q}_1 \geq \epsilon$, then \textsc{Identity-test}  returns \texttt{diff} with probability $\geq 1/30$.

Recall that there are $\frac{16}{\epsilon}$ tuples.
Also observe that all the three tests 
inside \textsc{Identity-test} 
are called with error parameter $\frac{\epsilon\delta}{48}$.
As a result if $p=q$, \textsc{Identity-test} outputs
\texttt{same} with probability $\geq 1- \frac{\epsilon \delta}{48} \cdot 3 \cdot \frac{16}{\epsilon} = 1-\delta$.

We now show that if $\norm{p-q}_1 \geq \epsilon$, then the algorithm
outputs \texttt{diff} with probability $\geq 1/30$.
By Lemma~\ref{lem:pickgood}, with probability $\geq 1/5$ \textsc{Find-element}
returns an element $x$ such that $p(x)-q(x) \geq \beta p(x)$ and $\alpha$-heavy.
Partition $G_x$ into groups $\cH= H_1,H_2,\ldots $s.t. for each group $H_j$, $p(x) \leq p(H_j) \leq 2p(x)$ and let $p^{\cH}_{G_x}$ and $q^{\cH}_{G_x}$ be the corresponding induced distributions. There are three possible cases. 
We show that for any $q$, at least one of the sub-routines
in \textsc{Identity-test} will output \texttt{diff} with
high probability.
\begin{enumerate}
 \item $|p(G_x)-q(G_x)|\ge \frac{\alpha\beta}{5}$.
 \item $|p(G_x)-q(G_x)|< \frac{\alpha\beta}{5}$ and $\norm{p^{\cH}_{G_x}-q^{\cH}_{G_x}}_1 \ge \frac{\beta}{5}$.
 \item $|p(G_x)-q(G_x)|< \frac{\alpha\beta}{5}$ and $\norm{p^{\cH}_{G_x}-q^{\cH}_{G_x}}_1<\frac{\beta}{5}$.
\end{enumerate}
If $|p(G_x)-q(G_x)|\ge \frac{\alpha\beta}{5}$, then 
chi-squared distance between $p^{\{G_x,G^c_x\}}$ and $q^{\{G_x,G^c_x\}}$
is $\geq  \bigl(\frac{\alpha\beta}{5}\bigr)^2$ and hence
\textsc{Test-equal}$\bigl((\alpha\beta/5)^2,\frac{\epsilon\delta}{48},p^{\{G_x,G^c_x\}},q^{\{G_x,G^c_x\}}\bigr)$ (step $2c$)
outputs \texttt{diff} with probability $>1-\frac{\epsilon\delta}{48}$.

If $|p(G_x)-q(G_x)| < \frac{\alpha\beta}{5}$ and $\norm{p^{\cH}_{G_x}-q^{\cH}_{G_x}}_1 \ge \frac{\beta}{5}$, then 
by Lemma~\ref{lem:uniformidentitytesting}
\textsc{Near-uniform-identity-test}$(\frac{\beta}{5},\frac{\epsilon\delta}{48},p^{\cH}_{G_x},q^{\cH}_{G_x})$
outputs \texttt{diff} with probability $>\frac15-\frac{\epsilon\delta}{48}>\frac16$.

If  $|p(G_x)-q(G_x)|< \frac{\alpha\beta}{5}$ and $\norm{p^{\cH}_{G_x}-q^{\cH}_{G_x}}_1<\frac{\beta}{5}$,
\begin{align*}
 \sum_{y\in \cH} p^{\cH}_{G_x}(y) \II \left[\frac{p(y)-q(y)}{p(y)}>\frac45\beta\right]
 &\le
 \frac{1}{p(G_x)} \sum_{y\in \cH} p(y) \II \left[p(y)-q(y)>\frac45\beta p(y)\right]\\
 &\le \frac{5}{4\beta p(G_x)} \sum_{y\in \cH} |p(y)-q(y)|\\
 &=  \frac{5}{4\beta} \sum_{y\in \cH} \lV \frac{p(y)}{p(G_x)}-\frac{q(y)}{p(G_x)}+\frac{q(y)}{q(G_x)}-\frac{q(y)}{q(G_x)} \rV\\
 &\stackrel{(a)}{\le}  \frac{5}{4\beta} \left(\sum_{y\in \cH} \lV\frac{p(y)}{p(G_x)}-\frac{q(y)}{q(G_x)}\rV+  \sum_{y\in \cH} q(y)\lV\frac{1}{p(G_x)}-\frac{1}{q(G_x)}\rV \right) 
\\
 &\stackrel{(b)}{\le}  \frac{5}{4\beta} \left(\frac{\beta}{5}+
 q(G_x)\frac{|p(G_x)-q(G_x)|}{p(G_x)q(G_x)} \right) \\
 &\stackrel{(c)}{\le}  \frac{5}{4\beta} \left(\frac{\beta}{5}+ \frac{\beta}{5} \right) \\
 &\le \frac12.
\end{align*}
$(a)$ follows from triangle inequality. $(b)$ follows from the fact that $\norm{p^{\cH}_{G_x} - q^{\cH}_{G_x}}_1 \leq \frac{\beta}{5}$.
$p(G_x) \geq \alpha$ and $p(G_x) -q(G_x) \leq \frac{\alpha \beta}{5}$ and hence $(c)$.
Therefore, for a random sample $y$ from $p^{\cH}_{G_x}$, with probability $\geq 1/2$,
$\frac{p(y)-q(y)}{p(y)}\le \frac{4\beta}5$. Let 
$\frac{q(y)}{p(y)} = \beta' \geq 1- \frac{4\beta}{5}$
and furthermore $\frac{q(x)}{p(x)} = \beta'' \leq 1-\beta$.
Hence $\beta'-\beta'' \geq \frac{\beta}{5}$.
Thus similar to the proof of Lemma~\ref{lem:chilow}, the chi-squared distance between $p_{\{x,y\}}$ and $q_{\{x,y\}}$ can be lower bounded by
\begin{align*}
& \geq 
\frac{(p(x)q(y)-q(x)p(y))^2}
{[p(x)(q(x)+q(y)) + q(x)(p(x)+p(y))][p(y)(q(x)+q(y)) + q(y)(p(x)+p(y))]} \\
& \stackrel{(a_1)}{\geq}
\frac{(\beta'-\beta'')^2 p^2(x) p^2(y)}
{[p(x)(q(x)+q(y)) + q(x)(p(x)+p(y))][p(y)(q(x)+q(y)) + q(y)(p(x)+p(y))]} \\
& \stackrel{(a_2)}{\geq} \frac{(\beta'-\beta'')^2}{\max^2(1,\beta',\beta'')}
\frac{p(x)p(y)}{4(p(x)+p(y))^2} \\
& \stackrel{(b)}{\geq} \frac{(\beta'-\beta'')^2}{18\max^2(1,\beta',\beta'')} \\
& \stackrel{(c)}{\geq} \frac{\beta^2}{1800}.
\end{align*}
$(a_1),(a_2)$ follow by substituting $q(x) = \beta' p(x)$ and $q(y) = \beta'' p(y)$. $(b)$ follows from $p(x) \leq 2 p(y)$ and $p(y) \leq 2p(x)$.$\beta'' \leq 1$ and $\beta'-\beta'' \geq \frac{\beta}{5}$
and hence the RHS in $(b)$ is minimized by $\beta' = 1+\frac{\beta}{5}$ and $\beta'' = 1$. For these values of $\beta',\beta''$,
$\max(1,\beta',\beta'') \leq 2$ and hence $(c)$.
Thus  \textsc{Test-equal} outputs \texttt{diff} with probability $> 1-\frac{\epsilon\delta}{48}$.

If $\norm{p-q}_1\ge \epsilon$, then by Lemma~\ref{lem:pickgood} step $1$ picks a tuple $(x,\beta,\alpha)$ such that $p(x)-q(x) \geq p(x)\beta$ with probability at least $\frac15$.
Conditioned on this event, for
the three cases discussed above the minimum probability of outputting
\texttt{diff} is $\frac16$
and every $p,q$ falls into one of the three categories.
Hence with probability $> \frac1{30}$ \textsc{Identity-test}
outputs \texttt{diff}.

We now compute the sample complexity of \textsc{Identity-test}.
Step $1$ of the algorithm uses $16/\epsilon$ samples. 
For every tuple $(x,\beta,\alpha)$, step $2(c)$ of the algorithm
uses $\cO \left(\frac{1}{\beta^2} \log \frac{1}{\delta \epsilon} \right)$ samples. Summing over all tuples yields a sample complexity of
\[
\sum^{16/\epsilon}_{j=1} \cO \left( \frac{1}{\beta^2_j} \log \frac{1}{\delta_j \epsilon} \right) = \cO \left(\frac{1}{\epsilon^2}  \log \frac{1}{\delta\epsilon}\right)
\]
For the different tuples \textsc{Test-equal}$\bigl(\frac{\alpha\beta}{5})^2,\frac{\epsilon\delta}{30},p^{\{G_x,G^c_x\}},p^{\{G_x,G^c_x\}}\bigr)$
can reuse samples and as $\alpha \beta = \Omega(\epsilon/(\log 1/\epsilon))$, it uses a total of
$\Theta\left( \frac{1}{\epsilon^2} \log^2 \frac{1}{\epsilon} \log \frac{1}{\delta\epsilon} \right)$. samples.

Furthermore, \textsc{Near-uniform-identity-test} uses $\cO\left( \frac{1}{\beta^2} \log \frac{1}{\epsilon \delta} \right)$ samples. Summing over all tuples, the sample complexity is
$\cO \left(\frac{1}{\epsilon^2}  \log \frac{1}{\delta\epsilon}\right)$.
Summing over all the three cases, the sample complexity of the algorithm is $\cO \left(\frac{1}{\epsilon^2}  \log^2 \frac{1}{\epsilon} \log \frac{1}{\delta\epsilon}\right)$.
\ignore{
\textsc{Near-uniform-identity-test}$(\frac{\beta}{5},\frac{\epsilon\epsilon}{30},p^{\cG}_S,q^{\cG}_S)$ and
\textsc{Test-equal}$(\frac{\beta}{25},\frac{\epsilon\epsilon}{30},p_{\{x,y\}},q_{\{x,y\}})$.
Summing over all three cases, the total number of samples \textsc{Identity-test} uses is
$\Theta\left( \frac{1}{\epsilon^2} \log^2 \frac{1}{\epsilon} \log \frac{1}{\epsilon\epsilon} \right)$.
}
\qedhere

\section{Closeness testing proofs}
\label{app:closeness}
\ignore{
\subsection{Finding good $i$}
We define $G_i = \{j: p(j)+q(j) \leq  p(i) + q(i)\}$.
As stated in Section~\ref{sec:closeness}, we are interested
in elements $i$ such that $p(i) > q(i)$ and $p(G_i) < q(G_i)$.
To this end,
we need a notion of approximable elements.
\begin{Definition}
For a pair of distributions $p,q$, 
element $i$ is $\beta$-approximable,
if
\[
\left \lvert \frac{p(i) - q(i)}{p(i)+q(i)} - \frac{p(G_i) - q(G_i)}{p(G_i) +q(G_i)} \right  \rvert \geq \beta.
\]
\end{Definition}
Before proceeding to prove how to find $\beta$-approximable elements.
We first show that
\begin{Lemma}
\label{lem:exp_approx}
\[
\sum_{i} \frac{p(i)+q(i)}{2} \left \lvert \frac{p(i) - q(i)}{p(i)+q(i)} - \frac{p(G_i) - q(G_i)}{p(G_i) +q(G_i)} \right  \rvert \geq \frac{\epsilon}{2}.
\]
\end{Lemma}
}
\subsection{Proof of Lemma~\ref{lem:exp_approx}}
\label{app:exp_approx}
Recall that $G_i = \{j: j \geq i\}$. Let $r_i = \frac{p(i)+q(i)}{2} $ and  $s_i = \frac{ p(i) - q(i)}{2}$.
We will use the following properties:
$\sum^k_{i=1} r_i = 1$, $\sum^k_{i=1} s_i = 0$, and $\sum^k_{i=1} |s_i| \geq  \frac{\epsilon}{2}$.
We will show that 
\[
\sum^k_{i=1} r_i \lV \frac{s_i}{r_i} - \frac{\sum^k_{j=i} s_j}{\sum^k_{j=i} r_j} \rV \geq \frac{\epsilon}{4}.
\]
 We show that 
\[
\sum^k_{i=1} r_i \lV \frac{s_i}{r_i} - \frac{\sum^k_{j=i} s_j}{\sum^k_{j=i} r_j} \rV 
\geq \frac{|s_1| + |s_2| - |s_1+s_2|}{2} + (r_1 +r_2) \lV \frac{s_1+s_2}{r_1+r_2} -  \frac{\sum^k_{j=1} s_j}{\sum^k_{j=1} r_j} \rV+ \sum^k_{i=3} r_i \lV \frac{s_i}{r_i} - \frac{\sum^k_{j=i} s_j}{\sum^k_{j=i} r_j} \rV.
\]
Thus reducing the problem from $k$ indices to $k-1$ indices with $s_1,s_2,\ldots s_k$ going to $s_1+s_2,s_3,\ldots s_k$ and $r_1,r_2,\ldots r_k$ going to $r_1+r+2,r_3,r_4,\ldots r_k$. Continuing similarly we can reduce the $k-1$ indices to $k-2$ indices with terms $s_1+s_2+s_3,s_4\ldots s_k$ and $r_1+r_2+r_3,r_4\ldots r_k$ and so on.
Telescopically adding the sum 
\begin{align*}
\sum^k_{i=1} r_i \lV \frac{s_i}{r_i} - \frac{\sum^k_{j=i} s_j}{\sum^k_{j=i} r_j} \rV  
& \geq \frac{|s_1| + |s_2| - |s_1+s_2|}{2} + \frac{|s_1+s_2|+|s_3| - |s_1+s_2+s_3|}{2}+ \ldots \\
& = \frac{\sum^k_{i=1} |s_i|}{2} \geq \frac{\epsilon}{4},
\end{align*}
where the last equality follows from the fact that $\sum^k_{i=1} s_i = 0$. To prove the required inductive step, it suffices to show
\begin{align*}
\sum^2_{i=1} r_i \lV \frac{s_i}{r_i} - \frac{\sum^k_{j=i} s_j}{\sum^k_{j=i} r_j} \rV
& \geq \frac{|s_1| + |s_2| - |s_1+s_2|}{2} + (r_1 +r_2) \lV \frac{s_1+s_2}{r_1+r_2} -  \frac{\sum^k_{j=1} s_j}{\sum^k_{j=1} r_j} \rV \\
& \geq \frac{|s_1|+|s_2|+|s_1+s_2|}{2},
\end{align*}
where the last inequality follows from the fact that $\sum^k_{i=1} s_i = 0$.
Rewriting the left hand side using the fact that $\sum^k_{i=1} s_i = 0$
\[
|s_1| + r_2 \lV \frac{s_2}{r_2} + \frac{s_1}{r_2+r'_3} \rV,
\]
where $r'_3 = \sum^k_{j=3} r_j$.
Thus it suffices to show 
\[
|s_1| + r_2 \lV \frac{s_2}{r_2} + \frac{s_1}{r_2+r'_3} \rV 
\geq  \frac{|s_1|+|s_2|+|s_1+s_2|}{2}. 
\]
We prove it by considering three sub-cases: $s_1$, $s_2$ have the same sign,
$s_1,s_2$ have different signs but $|s_1|\geq |s_2|$, and $s_1,s_2$ have different signs but $|s_1| < |s_2|$.
If $s_1,s_2$ have the same sign, then 
\[
|s_1| + r_2 \lV \frac{s_2}{r_2} + \frac{s_1}{r_2+r'_3} \rV 
\geq
|s_1| + r_2 \lV\frac{s_2}{r_2} \rV 
= 
|s_1| + |s_2|
=
\frac{|s_1|+|s_2|+|s_1+s_2|}{2}.
\]
If $s_1$ and $s_2$ have different signs and $|s_1| \geq |s_2|$, 
then
\[
|s_1| + r_2 \lV \frac{s_2}{r_2} + \frac{s_1}{r_2+r'_3} \rV 
\geq
|s_1| 
=
\frac{|s_1| + |s_1|}{2}
= 
\frac{|s_1|+|s_2|+|s_1+s_2|}{2}.
\]
If $s_1$ and $s_2$ have different signs and $|s_1| < |s_2|$, then
\[
|s_1| + r_2 \lV \frac{s_2}{r_2} + \frac{s_1}{r_2+r'_3} \rV
\geq |s_1| + r_2 \lV \frac{s_2}{r_2} + \frac{s_1}{r_2} \rV
=|s_1| + |s_2+s_1| 
=  \frac{|s_1|+|s_2|+|s_1+s_2|}{2}.
\]
\qedhere
\ignore{
If $s_1 \geq 0, s_2 < 0$, then the value depends on $s'_3 \ed - s_1 -s_2$.
If $s'_3 \geq 0$, then note that $s_1+s_2+s_3 = 0$, thus $|s_2| \geq |s_1|$,
and hence
\[
|s_1| + r_2 \lV \frac{s_2}{r_2} + \frac{s_1}{r_2+r'_3} \rV
\geq |s_1| + r_2 \lV \frac{s_2}{r_2} + \frac{s_1}{r_2} \rV
\geq |s_1| + |s_2+s_1| 
=  \frac{|s_1|+|s_2|+|s_1+s_2|}{2},
\]
the last part follows from $|s_2| = |s_1|+|s_3|$.
If $s_1 < 0, s_2 \geq 0$, then the value depends on $s'_3 \ed - s_1 -s_2$.
If $s_3 \geq 0$, then note that $|s_1| = |s_2|+|s_3|$,
thus
\[
|s_1| + r_2 \lV \frac{s_2}{r_2} + \frac{s_1}{r_2+r'_3} \rV
\geq |s_1| =  \frac{|s_1|+|s_2|+|s_1+s_2|}{2}.
\]
Thus proved.}
\ignore{ Similar to \textsc{Find-far}, we define the following algorithms
that finds $\beta$-approximable, $\alpha$-far elements.}
\ignore{
\begin{center}
\fbox{\begin{minipage}{1.0\textwidth}
Algorithm \textsc{Find-approximable} \newline
\textbf{Input:} $\epsilon$, distributions $p,q$. \newline
\textbf{Parameters:} Let $\epsilon' = \frac{\epsilon}{256 \log (4/\epsilon)}$.
$m' = 4/\epsilon'$, $\beta'_j = j\epsilon'/4 $, $\alpha'_j = 1/4(j \log (4/\epsilon')) $. 
\begin{enumerate}
\item Draw $m'$ independent samples $x_1,x_2\ldots x_{m'}$
from the first distribution and $m'$ samples $y_{1},y_{2},\ldots y_{m'}$ from the second distribution.
\item 
Output tuples \newline
$(x_1,\beta'_1/2,\alpha'_1,p), (x_2, \beta'_2/2,\alpha'_2,p),\ldots ,
(x_{m'},\beta'_{m'}/2,\alpha'_{m'},p)$ and \newline 
 $(y_1,\beta'_1/2, \alpha'_1,q), (y_2,\beta'_2/2,\alpha'_2,q), \ldots
y_{m'},\beta'_{m'}/2,\alpha'_{m'},q)$.
\end{enumerate}
\end{minipage}}
\end{center} 
}
\ignore{
\begin{Lemma}
\label{lem:find_approx}
If $\norm{p-q}_1\geq \epsilon$, then with probability $\geq 1/5$, \textsc{Find-element} returns a tuple $(x,\beta,\alpha,r)$
such that $x$ is $\alpha$-heavy and $\beta$-approximable with respect to $r$ for either $r=p$ or $q$.
\end{Lemma}
\begin{proof}
Recall that $m = 4/\epsilon$, $\beta_j = j\epsilon/4 $, $\alpha_j = 1/4(j \log (4/\epsilon)) $.
Without loss of generality, assume that $p_1\geq p_2\geq \ldots \geq p_k$. 
Note that this is a proof technique and the algorithm does not use this fact.
Consider the element $i$ with that largest value of $\beta_j \alpha_j$ 
such that it is $\beta_j$-far, $\alpha_j$-heavy, and not $\beta_j/2$-
approximable, for some $j$.  
If no such $i$, exists, every element $i$ 
that is $\beta_j$-far is also $\beta_j/2$-approximable. 
Thus by Lemma~\ref{lem:find_far}, the algorithm outputs an element $x_j$ such that it is $\beta_j/2$-approximable and $\alpha_j$-heavy.
Since $\beta_j >\beta'_j$ and $\alpha_j \geq \alpha'_j$. The lemma holds.
Suppose such $i$ exists, then
\[
\frac{p(G_i) - q(G_i)}{p(G_i)} \geq \frac{\beta_j}{2},
\]
and $p(G_i) \geq \alpha_j$.
However 
\[
\frac{p(G_1) - q(G_1)}{p(G_1)} =0.
\]
Thus there is a $i^* < i $ such that 
\[
\frac{p(G_{i^*}) - q(G_{i^*})}{p(G_{i^*})} \leq \frac{\beta_j}{4},
\]
Consider the largest $i^* < i$ such that the above inequality holds.
Let $A = G_{i^*}\setminus G_i$. Then,
\begin{align*}
p(G_{i^*}) - q(G_{i^*}) 
& = p(A) - q(A) + p(G_i) - q(G_i) \\
& \geq p(A) - q(A) + p(G_i)\frac{\beta_j}{2}.
\end{align*}
Furthermore,
\begin{align*}
p(G_{i^*}) - q(G_{i^*}) 
& \leq p(G_{i^*})\frac{\beta_j}{4} \\
& \leq p(A)\frac{\beta_j}{4} + p(G_i)\frac{\beta_j}{4}.
\end{align*}
Combining the two, we get
\[
q(A) - p(A) \left( 1- \frac{\beta_j}{4}\right) \geq p(G_i)\frac{\beta_j}{4} \geq \frac{\alpha_j\beta_j}{4} \geq \frac{\epsilon}{64 \log (4/\epsilon)}.
\]
Let $B = \{i_1 \in A : q(i_1) \geq p(i_1) (1-\beta_j/4)\}$.
Clearly,
\[
\frac{p(B)\beta_j}{4} + \sum_{i_1 \in B} (q(i_1) - p(i_1)) \geq \frac{\epsilon}{128 \log (4/\epsilon)}.
\]
Hence either $\sum_{i_1 \in B} (q(i_1) - p(i_1)) \geq \frac{\epsilon}{256 \log (4/\epsilon)}$
or $\frac{p(B)\beta_j}{4}  \geq \frac{\epsilon}{256 \log (4/\epsilon)}$
If the first inequality is satisfied, it is same as the condition in Lemma~\ref{lem:find_far}, 
for \textsc{Find-far} to return 
a good element with $\epsilon$ replaced by $\frac{\epsilon}{256 \log (4/\epsilon)}$.
Thus by Lemma~\ref{lem:find_far}, one of the pairs satisfy $q(x) - p(x) \geq q(x) \beta'_j \geq p(x) \beta'_j$.
For all such $i_1$ note that
\[
p(G_{i_1}) - q(G_{i_1}) \geq \frac{\beta_j}{4} p(G_{i_1}) \geq 0.
\]
Thus the returned element $x$ is $\alpha'_j$-heavy w.r.t $q$ and $\beta'_j/2$-approximable.
If the second condition is satisfied
then
\[
p(B) \geq \frac{\epsilon}{256 \beta_j \log (4/\epsilon)}.
\]
Thus there exists a $j$ and a set $B$ such that all elements in that set are $\beta_j$-approximable
and $\alpha_j$ heavy.
The probability that an element from $B$
appearing between $j' = 128j\log(4/\epsilon)$ and $j' = 384j\log(4/\epsilon)$ is $\geq 1-e^{-p(B)(384j\log(4/\epsilon)-128j\log(4/\epsilon))} \geq 1-e^{-1} \geq \frac{1}{2}$.
Let $t = 256 \log(4/\epsilon)$.
Thus one of the $x_j'$ in this range is  $\beta_{j'/t}$-approximable and $\alpha_{j'/t}$-heavy.
Since $\beta_{j'/t} \geq \beta'_{j'}$ and $\alpha_{j'/t} \geq \alpha'_{j'}$, the lemma holds.
\end{proof}
}
\ignore{We thus has a pair such that $(x,\beta,\alpha)$ such that $x$ is $\beta$-approximable
and $\alpha$-heavy. Note that since it is $\beta$-approximable, it is not the smallest element
w.r.t. the corresponding distribution.}
\subsection{Proof of Lemma~\ref{lem:reaches}}
\label{app:reaches}
We prove this lemma using several smaller sub-results. We first state a concentration result, which follows from Bernstein's inequality.
\begin{Lemma}
\label{lem:concentrates}
Consider a set $G$ such that  $\max_{j \in G} r(j) \leq r_{\max}$. Consider set $S$ formed by selecting each element from $G$ independently and uniformly randomly with probability $ r_0$, then
\[
E[r(S)] = r_0 r(G),
\]
and with probability $\geq 1-2\delta$,
\[
|r(S) - E[r(S)]| \leq \sqrt{2r_0r_{\max} r(G) \log \frac{1}{\delta}} + r_{\max} 
\log \frac{1}{\delta}.
\]
Furthermore
\[
E[|S|] = r_0 |G|,
\]
and with probability $\geq 1-2\delta$,
\[
||S| - r_0|G|| \leq \sqrt{2r_0|G| \log \frac{1}{\delta}} + 
\log \frac{1}{\delta} .
\]
\end{Lemma}
\ignore{
\subsection{Description of the algorithm}
In the rest of the algorithm we use $r$ to denote the averaged distribution \ie 
$r = (p+q)/2$.
The main algorithm is given in \textsc{Closeness-test}. It first finds a set of tuples 
and runs the main subroutine~\textsc{Assisted-closeness-test}. If the underlying tuple is good, then \textsc{Test-closeness-given-element} outputs the underlying hypothesis correctly.
\begin{center}
\fbox{\begin{minipage}{1.0\textwidth}
Algorithm \textsc{Closeness-test} \newline
\textbf{Input:} $\epsilon$, oracles $p,q$.
\textbf{parameters:} 
\begin{enumerate}
\item 
Generate a set of tuples using \textsc{Find-elements} using the distribution $r$.
\item 
For every tuple, run \textsc{Test-closeness-given-element}.
\item
If any of them returned \texttt{diff}, output \texttt{diff},
else \texttt{same}.
\end{enumerate}
\end{minipage}}
\end{center}
The main sub-routine for the algorithm is the binary search for estimating the value of 
$\frac{r(i)}{r(G_i)}$. At each step the algorithm finds sets $S_1,S_2\ldots S_m$.
It first prunes the sets to remove all heavy elements using $\textsc{Prune-set}$.
If $r_{\guess}$ is close to the underlying $\frac{r(i)}{r(G_i)}$,
then we show that for at least one of the underlying sets $S$
has a subset which we can use to differentiate between the two hypothesis.
The algorithm also has access to a comparator \textsc{Test-heavy}, that outputs 
\texttt{heavy} if $r_{\guess}$ is much higher than the actual value 
and outputs \texttt{light} if $r_{\guess}$ is much smaller than the actual value.
Thus the algorithm conducts a binary search over all possible values of $\frac{r(i)}{r(G_i)}$,
resulting a sample complexity of $\log \log k$.
\begin{center}
\fbox{\begin{minipage}{1.0\textwidth}
Algorithm \textsc{Test-closeness-given-element} \newline
\textbf{Input:} Oracles $p,q$, tuple $(i,\beta,\alpha)$.
\newline
\textbf{Parameters:} $n = $, $n_1 = $.
If $r =p$ run the algorithm as below. If $r=q$, replace $p$ by $q$
in all the subroutines and run the algorithm.
\begin{enumerate}
\item 
Initialize $\log r_{\guess} = -\log k$.
Set $\lowvalue = 0$ and $\highvalue = 2 \log k$. Do $2 \log \log k$ times:
\begin{enumerate}
\item Create $m$ sets $S_1,S_2,\ldots S_m$ independently,
by  keeping elements $\{1,2,\ldots,k\} \setminus \{i\}$.
in set each $S$ with probability $r_{\guess}$.
\item Prune all the sets $S_1,S_2,\ldots S_m$ using \textsc{Prune-sets}.
\item
For each set $S$ do:
\begin{enumerate}
\item Take $n$ samples from $r_{S\cup\{i\}}$ 
and 
for all  seen elements $j$, run \textsc{Test-equal} $\left(p_{\{i,j\}}, q_{\{i,j\}}, \epsilon',\frac{\beta^4 \alpha^3}{100\gamma^{3/2}} \right)$.
\item
Let $\cS=\{ \{i\},S\}$. Run \textsc{Test-equal}$\left( p^{\cS}_{S\cup \{i\}}, p^{\cS}_{S\cup \{i\}},
\epsilon', \frac{\alpha \beta}{10\sqrt{\gamma}} \right)$.
\end{enumerate}
\item Run \textsc{Test-heavy}$(S_1,i,\alpha)$.
\begin{enumerate}
\item If output is \texttt{heavy}, update $\high = \log r_{\guess}$ and 
$\log r_{\guess} = (\log r_{\guess}+ \low)/2$.
\item If output is \texttt{light}, update
$\low = \log r_{\guess}$ and
$\log r_{\guess} = (\log r_{\guess}+ \high)/2$.
\end{enumerate}
\end{enumerate}
\end{enumerate}
\end{minipage}}
\end{center} 
\begin{center}
\fbox{\begin{minipage}{1.0\textwidth}
Algorithm \textsc{Prune-set} \newline
\textbf{Input:} set $S$, element $i$, $\alpha$.
\newline
\textbf{Parameters:} $\epsilon' = \frac{\epsilon}{10 m \log \log k}$, $n_1 = 4\log \frac{\gamma}{2\epsilon'\alpha} 
 \left( \frac{\gamma}{2\alpha}\log \frac{\gamma}{2\alpha}+2\log\frac{1}{\epsilon'}\log\log \frac{1}{\epsilon'}\right)$, $n_2 = \cO(\log \log \log k + \log \frac{1}{\epsilon\epsilon})$.
 \newline
 Repeat $n_1$ times:
\begin{enumerate}
 \item 
Obtain a sample $j$ from $r_S$:
\begin{enumerate}
\item
Sample $n_2$ times from $r_{\{j,i\}}$ and if $n(j) \geq 3n/4$, remove $j$ from set $S$.
\end{enumerate}
\end{enumerate}
\end{minipage}}
\end{center} 
\begin{center}
\fbox{\begin{minipage}{1.0\textwidth}
Algorithm \textsc{Test-heavy} \newline
\textbf{Input:} set $S$, element $i$, $\alpha$.
\newline
\textbf{Parameters:} $n= \cO(\frac{1}{\gamma^2} \log \log k)$.
\begin{enumerate}
\item
Obtain $n$ samples from the entire set $S$.
If $n(i) < \frac{4n}{\gamma}$, then output \texttt{heavy}, else output \texttt{light}.
\end{enumerate}
\end{minipage}}
\end{center} 
}
\ignore{
We now state few state few properties of the algorithm and they follow from the properties
of binary search.}
\subsubsection{Results on \textsc{Prune-set}}
We now show that with high probability~\textsc{Prune-set} never removes an element with probability $\leq 2r(i)$.
\begin{Lemma}
\label{lem:two_remains}
The probability that the algorithm removes
an element $j$ from the set $S$ such that 
$r(j) \leq 2r(i)$  during step $2$ of \textsc{Binary-search} is $\leq \delta/5$.
\end{Lemma}
\begin{proof}
If $r(j)<2r(i)$, then  $\frac{r(j)}{r(j)+r(i)} \leq \frac{2}{3}$. 
Applying Chernoff bound,
\begin{equation*}
\Pr\left(n(j) \geq \frac{3n_2}{4}\right) \leq e^{-n_2/72}.
\end{equation*}
Since the algorithm uses this step
 no more than $\cO(n_1\log \log k)$ times, 
the total error probability is less than $\cO(n_1 \log \log k\cdot  e^{-n_2/72})$.  Since $n_1$ is $\text{poly}(\log \log \log k,\epsilon^{-1}, \log \delta^{-1})$
and 
$n_2 = \cO(\log \log \log k + \log \frac{1}{\epsilon\delta})$,
the error probability is $\leq \delta/5$.
\end{proof}
We now show that \textsc{Prune-set} removes all elements with probability $\geq 4r(i)$ 
with high probability. Recall that $\delta'=\frac{\delta}{40m\log\log k}$.
\begin{Lemma}
\label{lem:four_no_remain}
If element $i$ is $\alpha$-heavy, $\beta$-approximable and  $r_{\guess} \leq \frac{\gamma}{\beta} \frac{r(i)}{r(G_i)} $,
 then \textsc{Prune-set} removes all elements such that $r(j) > 4 r(i)$ during all calls of step $2$ of \textsc{Binary-search} with probability $\geq 1 -\frac{\delta}{5}$.
\end{Lemma}
\begin{proof}
Let $A = \{j: r(j) \leq 4 r(i)\}$ and $S' = S \cap A$.
By Lemma~\ref{lem:concentrates}, with probability $\geq 1-2\delta'$
\[
r(S') \leq r_{\guess}r(A)+\sqrt{8r_{\guess}r(i)r(A) \log\frac{1}{\delta'}} + 4r(i) \log\frac{1}{\delta'}
\le  2r_{\guess}+ 8r(i) \log\frac{1}{\delta'},
\]
where the last inequality follows from the identity $\sqrt{2ab} \leq a+b$.
Observe that $|A^c|\leq \frac{1}{4r(i)}$. Let $S''= S\setminus S'$.
By Lemma~\ref{lem:concentrates}, with probability $\geq 1-2\delta'$
\[
\nu \ed |S''| \leq r_{\guess}\frac{1}{4r(i)}+\sqrt{2r_{\guess}\frac{1}{4r(i)} \log \frac{1}{\delta'}}+\log \frac{1}{\delta'}
\le \frac{r_{\guess}}{2r(i)}+2\log \frac{1}{\delta'}.
\]
$S$ has $\nu$ elements with probability $> 4r(i)$.
\ignore{ and the total probability of other
elements is $r(S')$.
The lemma reduces to the number of samples we should take, $n_1$, to observe
all the heavy elements.}
Suppose we have observed $j$ of these  elements and removed them from $S$. 
There are $\nu-j$ of them left in $S$. After taking another $\eta$ samples 
from $S$, the probability of not observing a $(j+1)$th heavy element is
$<\left(r(S')/(r(S')+4r(i)(\nu-j))\right)^\eta$.
Therefore,
\[
\eta_j \ed \log \frac{\nu}{\delta'} \cdot \left(1+\frac{r(S')}{4r(i)(\nu-j)}\right) \geq \frac{\log \frac{\nu}{\delta'}}{\log  \left(1+\frac{4r(i)(\nu-j)}{r(S')}\right)}
\]
samples suffice to observe an element from $S''$ with probability $> 1-\frac{\delta'}{\nu}$. After observing the sample (call it $j$), similar to the proof of Lemma~\ref{lem:two_remains} it can be shown that   with probability 
$\geq 1-\delta'$, for samples from $r_{\{j,i\}}$, $n(j) \geq 3n_2/4$ and hence $j$ will be removed from $S$. Thus to remove all $\nu$ elements of probability $>4r(i)$, we need to repeat this step
\[
n_1  = \sum^\nu_{j=1} \eta_j = \log \frac{\nu}{\delta'} \cdot \sum_{j=1}^{\nu} \left(1+\frac{r(S')}{4r(i)j}\right)\le
\log \frac{\nu}{\delta'} \cdot \left(\nu+\frac{r(S')}{4r(i)} \log \nu \right)
\]
times.
Substituting $r(S')$ and $\nu$ in the RHS and simplifying we have
\[
 n_1 \le 4\log \frac{\gamma}{2\delta'\alpha \beta} 
 \left( \frac{\gamma}{2\alpha \beta}\log \frac{\gamma}{2\alpha \beta}+2\log\frac{1}{\delta'}\log\log \frac{1}{\delta'}\right).
\]
 By the union bound, total error probability is $\leq \delta'$. Since the number of calls to \textsc{Prune-set} is at most $\log \log k$ during step $2$ of the algorithm, the error is at most $ \log \log k \cdot 2\delta' \leq \delta/5$ and  the lemma follows from the union bound.
\end{proof}
\subsubsection{Proof of Lemma~\ref{lem:reaches}}
The proof of Lemma~\ref{lem:reaches} follows from the following two sub-lemmas. In Lemma~\ref{lem:heavy}, we show that if $r_{\guess} \geq \gamma \frac{r(i)}{r(G_i)}$ then step $4$ will return \texttt{heavy}, and if $r_{\guess} \leq \frac{1}{\gamma} \frac{r(i)}{r(G_i)}$ hence the algorithm outputs \texttt{light} with high probability. Since we have $\log \log k$ iterations and $\frac{1}{k} \leq \frac{r(i)}{r(G_i)} \leq 1$, we reach $\frac{r(i)}{\gamma r(G_i)} \leq r_{\guess} \leq \frac{\gamma r(i)}{\beta r(G_i)}$ at some point of the algorithm.
\ignore{
Once the algorithm prunes the sets correctly, we now show that the algorithm correctly outputs \texttt{heavy} or \texttt{light} with high probability.
We set $\gamma$ later, but for now we assume $\gamma \geq 100 \log \frac{\log \log k}{\delta}$.}
\begin{Lemma}
\label{lem:heavy}
If $r_{\guess} > \frac{\gamma}{\beta} \frac{r(i)}{r(G_i)}$, $i$ is $\alpha$ heavy, $\beta$-approximable, 
and \textsc{Prune-set} has removed 
 none of the elements with probability $\leq 2r(i)$, then
with probability $\geq 1- 4\delta'$, step $4$ outputs \texttt{heavy}.
\end{Lemma}
\begin{proof}
Let $G'_i = G_i \setminus \{i\}$.
Since $i$ is $\alpha$ heavy and $\beta$-approximable, by convexity
\[
\frac{r(G'_i)}{r(G_i)} 
\left \lvert \frac{p(i) - q(i)}{p(i)+q(i)} - \frac{p(G'_i) - q(G'_i)}{p(G'_i) +q(G'_i)} \right  \rvert
 \geq 
\left \lvert \frac{p(i) - q(i)}{p(i)+q(i)} - \frac{p(G_i) - q(G_i)}{p(G_i) +q(G_i)} \right  \rvert \geq \beta.
\]
Hence $r(G'_i) \geq \beta r(G_i)/2$.
By assumption,
 all the elements with probability $<2r(i)$ in set $S$
will remain after pruning. 
Thus all the elements in set $S$ from $G'_i$ remains after pruning.
Let $S' = G'_i \cap S$.
By Lemma~\ref{lem:concentrates} with 
$G = G'_i$, $r_{\max} = r_i$, and $r_0 = r_{\guess}$
\begin{equation}
\label{eq:derivatives}
\Pr\left(r(S')\leq 
r_{\guess}r(G'_i)-\sqrt{2r_{\guess}r(i)r(G'_i)\log \frac{1}{\delta'}} - r(i)\log \frac{1}{\delta'}\right) <2\delta'.
\end{equation}
Taking derivatives, it can be shown that the slope of the RHS of the term inside parenthesis is $r(G'_i) - \sqrt{\frac{r(i)r(G'_i)\log \frac{1}{\delta'}}{2r_{\guess}}}$, which is positive for  $r_{\guess} \geq \frac{\gamma r(i)}{\beta r(G_i)}$. Thus the value is minimized at $r_{\guess} =  \frac{\gamma r(i)}{\beta r(G_i)}$ in the range $\left[  \frac{\gamma r(i)}{\beta r(G_i)}, \infty\right)$ and simplifying this lower bound using values of $\gamma,\beta$, we get
\[
r_{\guess}r(G'_i)-\sqrt{2r_{\guess}r(i)r(G'_i)\log \frac{1}{\delta'}} - r(i)\log \frac{1}{\delta'} \geq \frac{\gamma  r(i)}{4}.
\]
Since $\Pr(X <b) \leq \Pr(X <b+t)$, we have
\ignore{
By choosing $r_{\guess} =\frac{\gamma_1 r(i)}{r(G_i)}$ where $\gamma_1 > \gamma$,
\[
\Pr\left(r(S') \leq r(i)\bigl(\gamma_1-\sqrt{2\gamma_1\log \frac{1}{\delta}}-\log \frac{1}{\delta'}\bigr)\right)<2\delta'.
\]
Thus, if $\gamma_1>\gamma> 100 \log \frac{\log \log k}{\delta}$
and choosing $\delta' = \frac{\delta}{20\log \log k}$}
\[
 \Pr\left(r(S') \leq \frac{\gamma  r(i)}{4}\right) 
< 2 \delta'. \]
Hence with probability $\geq 1 - 2\delta'$,
\ignore{However, as we just showed, $\Pr\left(r(S') \geq \frac{\gamma r(i)}{2}\right) \leq \frac{\delta}{10 \log \log k}$ 
and consequently} $\frac{r(i)}{r(S)+r(i)} \leq \frac{r(i)}{r(S')}\leq \frac {4}{\gamma }$. By the Chernoff bound, 
\[
\Pr\left(\frac{n(i)}{n_3} > \frac{5}{\gamma }\right) \leq \Pr\left(\frac{n(i)}{n_3} >\frac{r(i)}{r(S)+r(i)} + \frac{1}{\gamma }\right) \leq e^{-2n_3/\gamma^2}.
\]
\ignore{
and 
\[
 \leq e^{-4n_3/\gamma^2}.
\]}
Therefore, for 
$n_3 \geq \cO \left({\gamma^2 } \log \frac{\log \log k}{\delta} \right)$, step $3$ outputs \texttt{heavy} with probability $\geq  1- 2\delta'$. By the union bound
the total error probability $\leq 4\delta' $
\end{proof}
\begin{Lemma}
\label{lem:light}
If $r_{\guess} < \frac{r(i)}{\gamma}$ and \textsc{Prune-set} has removed all elements with probability $\geq 4r(i)$ and none of the elements with probability $\leq 2r(i)$, then with probability 
$\geq 1- 4\delta'$, 
step $3$, outputs \texttt{light}.
\end{Lemma}
\begin{proof}
The proof is similar to that of Lemma~\ref{lem:heavy}.
By assumption all the elements have probability $\leq 4r(i)$.
By Lemma~\ref{lem:concentrates},
\[
\Pr\left( r(S)>r(i)\left[\sqrt{8 \frac{ r_{\guess}}{r(i)} \log \frac{1}{\delta'}}+ \frac{ r_{\guess}}{r(i)}+4\log \frac{1}{\delta'}\right] \right) \leq 2\delta'.
\]
Similar to the analysis after Equation~\eqref{eq:derivatives}, taking derivatives it can be shown that the RHS of the term inside parenthesis
is maximized when $r_{\guess}= \frac{r(i)}{\gamma}$ for the range $[0, \frac{r(i)}{\gamma}]$.
Thus simplifying the above expression with this value of $r_{\guess}$ and the value of $\gamma$, with probability $\geq 1-2\delta'$, $r(S) \leq \gamma r(i)/10 $.
\ignore{Selecting $\gamma_1>\gamma>100 \frac{\log \log k}{\delta}$, 
$\delta' = \frac{\delta}{10 \log \log k}$, $r(S) \leq 4r(i) \log \frac{\log \log k}{\delta}$ with probability at least
 $1-\frac{\delta}{\log \log k}$.}  Thus with probability $\geq 1- 2\delta'$,
\[
\frac{r(i)}{r(S)+r(i)} \geq \frac{1}{1+\gamma/10} \geq \frac{6}{\gamma}.
\]
By the Chernoff bound
\[
\Pr \left( \frac{n_1(i)}{n_3} \leq \frac{5}{\gamma} \right) \leq \Pr \left( \frac{n_1(i)}{n_3} \leq \frac{r(i)}{r(S)+r(i)} - \frac{1}{\gamma} \right) \leq e^{-2n_3/\gamma^2}. 
\]
\ignore{
Choosing $\delta=\frac{1}{8\log \frac{\log \log k}{\delta}}$ and since $\frac{r(i)}{r(S)} \geq \frac{1}{4 \log \frac{\log \log k}{\delta}}$
\[
\Pr\left( \frac{n_1(i)}{n_3} >\frac{1}{8 \log \frac{\log \log k}{\delta}} \right) \leq e^{-n\frac{1}{(8\log \frac{\log \log k}{\delta})^2}}.
\] }
The lemma follows from the bound on $n_3$ and by the union bound total error probability $\leq 4\delta'$.
\end{proof}

Note that the conditions in Lemmas~\ref{lem:heavy} and~\ref{lem:light} hold with probability $\geq 1 - \frac{2\delta}{5}$ by Lemmas~\ref{lem:four_no_remain} and~\ref{lem:two_remains}. Furthermore, since we use all the steps at most $\log \log k$ times, by the union bound, the conclusion in Lemma~\ref{lem:reaches} fails with probability $\leq \frac{2\delta}{5} +  \log \log k \cdot 8\delta' = \frac{2\delta}{5} \leq \delta$.
\subsection{Proof of Theorem~\ref{thm:closeness}}
\label{app:thmcloseness}
For the ease of readability, we divide the proof into several sub-cases. We first show that if $p=q$, then the algorithm returns \texttt{same} with high probability. Recall that for notational simplicity we redefine $\delta' = \frac{\epsilon\delta}{32m( n_4+1) \log \log k} $. 
\begin{Lemma}
\label{lem:same}
If $p = q$, \textsc{Closeness test}  outputs \texttt{same} with error probability $\leq \delta$. 
\end{Lemma}
\begin{proof}
Note that the algorithm returns \texttt{diff}
only if any of the \textsc{Test-Equal}s return \texttt{diff}.
We call \textsc{Test-Equal} at most $\frac{16}{\epsilon} \cdot  \log \log k \cdot m  \cdot (n_4+1)$ times. The probability that
at any time it returns an error is $\leq \delta'$. 
Thus by the union bound total error probability is $\leq \delta' 16\log \log k \cdot m \cdot ( n_4+1)/\epsilon \leq \delta$.
\end{proof}
We now prove the result when $\norm{p-q}_1 \geq \epsilon$. We first state a lemma showing that \textsc{Prune-set} ensures that set $S$ does not have any elements $\geq 4 r(i)$. The proof is similar to that of Lemmas~\ref{lem:four_no_remain} and~\ref{lem:two_remains} and hence omitted.
\begin{Lemma}
\label{lem:prune_closeness}
If $i$ is $\alpha$-heavy and $\beta$-approximable, then at any call of step $2$ of \textsc{Assisted-closeness-test}, with probability $\geq 1- \frac{2\delta}{5}$, if $r_{\guess}\leq \frac{\gamma}{\beta} \frac{r(i)}{r(G_i)}$, then \textsc{Prune-test} never removes an element with probability $\leq 2 r(i)$ and removes all elements with probability $\geq 4 r(i)$.
\end{Lemma}
The proof when $\norm{p-q}_1 \geq \epsilon$ is divided into two parts based on the probability of certain events.
Let $\beta' = \frac{p(i)-q(i)}{p(i)+q(i)}$,
$\beta'' = \frac{\alpha\beta}{128\gamma \log \frac{128\gamma}{\beta^2}}$.
Let $D$ denote the event such that an element $j$ from $G^c_i$
with $\left \lvert\frac{p(j)-q(j)}{p(j)+q(j)} - \beta'\right \rvert \geq \beta''$ and $r(j) \leq 4r(i)$ gets included in $S$. 
We divide the proof in two cases when $\Pr(D) \geq \frac{ \alpha \beta^2}{128 \gamma}$ and $\Pr(D) <\frac{\alpha \beta^2}{128 \gamma}$.
\begin{Lemma}
\label{lem:diff_1}
Suppose $\norm{p-q}_1 \geq \epsilon$. If $i$ is $\alpha$-heavy and $\beta$-approximable, $\frac{r(i)}{\gamma} \leq r_{\guess} \leq \frac{\gamma}{\beta} \frac{r(i)}{r(G_i)}$, the conclusions in Lemma~\ref{lem:prune_closeness} hold
, and $\Pr(D) \geq \frac{\alpha \beta^2}{128 \gamma}$, then step $3 (a)$ of \textsc{Assisted-closeness-test}
returns \texttt{diff} with probability $\geq 1/5$.
\end{Lemma}
\begin{proof}
we then show that the following four events happen with high probability for at least one set $S \in \{S_1,S_2\ldots S_m\}$.
\begin{itemize}
\item $S$ includes a $j$ such that  $\left \lvert\frac{p(j)-q(j)}{p(j)+q(j)} - \beta'\right \rvert \geq \beta''$,$r(j) \leq 4r(i)$ , $j \notin G_i$.
\item $r(S) \leq r_{\guess} + 8 \sqrt{{r(i) }{r_{\guess}}}$.
\item $j$ appears when $S$ is sampled $n_4$  times.
\item \textsc{Test-equal} returns \texttt{diff}.
\end{itemize}
Clearly, if the above four events happen then the algorithm outputs \texttt{diff}. Thus to bound the error probability, we bound the error probability of each of the above four events and use union bound. Probability that at least one of the sets contain an element $j$ such that  $\left \lvert\frac{p(j)-q(j)}{p(j)+q(j)} - \beta'\right \rvert \geq \beta''$,$r(j) \leq 4r(i)$ , $j \notin G_i$ is 
\[
1 - \left( 1 - \Pr(D) \right)^{m} \geq 1 -  e^{-\Pr(D)m} \geq \frac{5}{6}.
\]
 Let $S' = \{j \in S: r(j) \leq 4 r(i)\}$. Observe that before pruning $\EE[r(S')] \leq r_{\guess}$ and $\Var(r(S')) \leq 4r(i) r_{\guess}$. Hence by the Chebyshev bound  with probability $\geq 1- 1/16$,
\begin{equation*}
r(S') \leq r_{\guess} + 8\sqrt{{r(i) }{r_{\guess}}} ,
\end{equation*}
After pruning, $r(S)$ contains only elements of probabilities from $S'$. Hence with probability $\geq 1- 1/16$,
$r(S) \leq r_{\guess} + 8 \sqrt{ r(i){r_{\guess}}}$.  Probability that this element does appear when sampled $n_4$ times is 
\[
1 - \left( 1- \frac{r(j)}{r(S)} \right)^{n_4} 
\geq 
1 - \left( 1- \frac{r(i)}{r_{\guess}(1+8\sqrt{r(i)/r_{\guess}})} \right)^{n_4}
\geq 
1 -  \left( 1 - \frac{\alpha \beta}{9\gamma} \right)^{n_4}
 \geq 
\frac{5}{6}.
\]
Since 
 $\left \lvert\frac{p(j)-q(j)}{p(j)+q(j)} - \beta'\right \rvert \geq \beta''$,$r(i) \leq r(j) \leq 4r(i)$ by  Lemma~\ref{lem:chilow}
the chi-squared distance is 
\begin{align*}
\geq (\beta'')^2 \frac{r(j)r(i)}{4(r(i)+r(j))^2} \geq \frac{(\beta'')^2}{25}.
\end{align*}
Thus by Lemma~\ref{lem:test_equal}, algorithm outputs \texttt{diff} with probability $1 - \delta'$.  By the union bound the total error probability $\leq 1/6 + 1/16 + +1/6 + \delta' \leq 4/5$.
\end{proof}
\begin{Lemma}
\label{lem:diff_2}
Suppose $\norm{p-q}_1 \geq \epsilon$. If $i$ is $\alpha$-heavy and $\beta$-approximable, $\frac{r(i)}{\gamma} \leq r_{\guess} \leq \frac{\gamma}{\beta} \frac{r(i)}{r(G_i)}$, the conclusions in Lemma~\ref{lem:prune_closeness} hold
, and $\Pr(D) < \frac{\alpha \beta^2}{128 \gamma}$, then step $3 (b)$ of \textsc{Assisted-closeness-test}
returns \texttt{diff} with probability $\geq 1/5$.
\end{Lemma}
We show that the following four events happen with high probability
for at least some set $S \in \{S_1,S_2,\ldots,S_m\}$.
Let $S' = S \cap G_i$ and $G'_i = G_i \setminus \{i\}$. 
Let 
\[
Z = r(S')\left \lvert \beta' - \frac{p(S')-q(S')}{p(S')+q(S')}\right \rvert = \frac{\lvert \beta'(p(S')+q(S')) - (p(S')-q(S')) \rvert}{2}.
\]
\begin{itemize}
\item 
$Z \geq r_{\guess}\lvert( \beta'(p(G'_i) + q(G'_i)) - p(G'_i)+ q(G'_i))\rvert /4$.
\item 
$r(S) \leq 8 (r(i) + r_{\guess}) \log \frac{128 \gamma}{\alpha\beta^2}$.
\item 
Event $D$ does not happen.
\item \textsc{Test-equal} outputs \texttt{diff}.
\end{itemize}
Clearly if all of the above events happen, then the test outputs \texttt{diff}. We now bound the error probability of each of the events and use union bound.
Since none of the elements in $S'$ undergo pruning, 
the value of $Z$ remains unchanged before and after pruning. Thus any concentration inequality for $Z$ remains the same after pruning.
We now compute the expectation and variance of $Z$ and use Paley Zigmund inequality.
\begin{align*}
\EE[Z] 
& = \EE[\lvert \beta'(p(S')+q(S')) - p(S')+q(S') \rvert]/2 \\
& \geq \lvert \EE( \beta'(p(S')+q(S')) - p(S')+q(S'))\rvert /2 \\
& = r_{\guess}\lvert( \beta'(p(G'_1) + q(G'_i)) - p(G'_i)+ q(G'_i))\rvert /2,
\end{align*}
where the inequality follows from convexity of $|\cdot|$ function.
Let $\mathbbm{1}(j,S')$ denote the event that $j \in S'$.
The variance is lower bounded as 
\begin{align*}
\Var(Z) 
& = \EE[Z^2] - (\EE[Z])^2 \\
& =  \EE[( \beta'(p(S')+q(S')) - p(S')+q(S'))^2]/4 
-  \EE^2[\lvert \beta'(p(S')+q(S')) - p(S')+q(S') \rvert]/4 \\
& \stackrel{(a)}{\leq}  \EE[( \beta'(p(S')+q(S')) - p(S')+q(S'))^2]/4
-  \EE^2[ \beta'(p(S')+q(S')) - p(S')+q(S')]/4 \\
& {=} \Var( \beta'(p(S')+q(S')) - p(S')+q(S'))/4 \\
& \stackrel{(b)}{=} \sum_{j \in G'_i}  \Var(\mathbbm{1}(j,S')( \beta'(p(j)+q(j)) - p(j)+q(j))^2/4\\
& {\leq} \sum_{j \in G'_i}  \EE[\mathbbm{1}(j,S')]( \beta'(p(j)+q(j)) - p(j)+q(j))^2/4\\
& = \sum_{j \in G'_i} r_{\guess}( \beta'(p(j)+q(j)) - p(j)+q(j))^2/4\\
& {\leq}  \max_{j' \in G'_i}| \beta'(p(j')+q(j')) - p(j')+q(j')| \cdot 
r_{\guess}  \sum_{j \in G'_i} |\beta'(p(j)+q(j)) - p(j)+q(j)|/4\\
& \stackrel{(c)}{\leq} 4 r(i) \cdot r_{\guess} r(G'_i).
\end{align*}
$(a)$ follows from the bound on expectation. $(b)$ follows from the
independence of events $\mathbbm{1}(j,S)$. $(c)$ follows from the fact that $p(j)+q(j) = 2r(j) \leq 2r(i)$, $|\beta'|\leq 1$ and $\sum_{i} r(j) \leq r(G'_i)$.
Hence by the Paley Zygmund inequality,
\begin{align*}
\Pr(Z \geq  r_{\guess}\lvert( \beta'(p(G'_1) + q(G'_i)) - p(G'_i)+ q(G'_i))\rvert /4) 
& \geq \Pr(Z \geq \EE[Z]/2) \\
& \geq \frac{1}{4} \frac{\EE^2[Z]}{\Var(Z) + \EE^2[Z]} \\
& \geq \frac{1}{4} \frac{\EE^2[Z]}{4 r(G'_i) r(i) r_{\guess} + \EE^2[Z]}.
\end{align*}
Since $i$ is $\beta$-approximable, by convexity
\[
\frac{r(G'_i)}{r(G_i)} 
\left \lvert \frac{p(i) - q(i)}{p(i)+q(i)} - \frac{p(G'_i) - q(G'_i)}{p(G'_i) +q(G'_i)} \right  \rvert
 \geq 
\left \lvert \frac{p(i) - q(i)}{p(i)+q(i)} - \frac{p(G_i) - q(G_i)}{p(G_i) +q(G_i)} \right  \rvert \geq \beta.
\]
Hence,
\[
r(G'_i) \left \lvert \frac{p(i) - q(i)}{p(i)+q(i)} - \frac{p(G'_i) - q(G'_i)}{p(G'_i) +q(G'_i)} \right  \rvert \geq r(G_i) \beta.
\]
Thus $\EE[Z] \geq r_{\guess} r(G_i) \beta$ and 
\begin{align}
\Pr(Z \geq  r_{\guess}\lvert( \beta'(p(G'_1) + q(G'_i)) - p(G'_i)+ q(G'_i))\rvert /4) 
& \geq \frac{1}{4} \frac{(r_{\guess} r(G_i) \beta)^2}{4r(i)r_{\guess} r(G'_i)+ (r_{\guess} r(G_i) \beta)^2} 
\nonumber \\
& \geq \frac{1}{4}  \frac{(r_{\guess} r(G_i) \beta)^2}{2\max(4r(i)r_{\guess} r(G'_i), (r_{\guess} r(G_i) \beta)^2)}
\nonumber \\ 
& = \frac{1}{8} \min \left(1, \frac{(r_{\guess} r(G_i) \beta)^2}{4r(i) r_{\guess} r(G'_i)} \right) 
\nonumber \\
& \geq \frac{r(G_i) \beta^2}{32\gamma} 
\nonumber \\
& \geq \frac{\alpha\beta^2}{32\gamma }.
\label{eq:zlower}
\end{align}
By Lemma~\ref{lem:concentrates}, with probability $\geq 1- \frac{\alpha\beta^2}{128 \gamma}$,
\begin{equation}
\label{eq:aux_lower3}
r(S) \leq r_{\guess} + \sqrt{8 r_{\guess}r(i) \log \frac{128 \gamma}{\alpha\beta^2}} + 4 r(i) \log \frac{128 \gamma}{\alpha\beta^2} \leq 8 (r(i)+r_{\guess})\log \frac{128\gamma}{\alpha\beta^2} .
\end{equation}
Let $S'' = S \setminus S'$. If event $D$ has not happened then for all elements $j \in S''$,
$ \lV \beta' - \frac{p(j) - q(j)}{p(j)+q(j)} \rV \leq \beta''$ and hence 
\begin{equation}
\label{eq:aux_lower2}
\left \lvert \frac{p(i)-q(i)}{p(i)+q(i)} - \frac{p(S'') - q(S'')}{p(S'')+q(S'')} \right  \rvert  \leq \beta''.
\end{equation}
Combining the above set of equations,
\begin{align*}
 \left \lvert \frac{p(i)-q(i)}{p(i)+q(i)} - \frac{p(S) - q(S)}{p(S)+q(S)} \right  \rvert 
 & \stackrel{(a)}{\geq}  \frac{r(S')}{r(S)}
\left \lvert \frac{p(i)-q(i)}{p(i)+q(i)} - \frac{p(S') - q(S')}{p(S')+q(S')} \right  \rvert 
- \frac{r(S'')}{r(S)} \left \lvert \frac{p(i)-q(i)}{p(i)+q(i)} - \frac{p(S'') - q(S'')}{p(S'')+q(S'')} \right  \rvert \\
& 
\stackrel{(b)}{\geq}  \frac{Z}{r(S)}
- \beta'' \\
& \stackrel{(c)}{\geq}  \frac{2r(G'_i)r_{\guess}}{4r(S)} \left(
\left \lvert\beta' - \frac{p(G'_i)-q(G'_i)}{p(G_i)+q(G_i)}\right \rvert  \right)
-  \beta''\\
& \geq \frac{r(G_i) r_{\guess} \beta}{2r(S)} - \beta''\\
& \stackrel{(d)}{\geq}  \frac{r(G_i) r_{\guess} \beta}{8r(S)}.
\end{align*}
$(a)$ follows from convexity and the fact that $|a+b| \geq |a| - |b|$, $(b)$ follows from Equation~\eqref{eq:aux_lower2}, and $(c)$ follows from Equation~\eqref{eq:zlower}.
$(d)$  follows from the fact that 
\begin{align*}
\beta'' = \frac{\alpha \beta}{128 \gamma  \log \frac{128\gamma}{\alpha\beta^2}} & \leq  \frac{\beta}{128 \log \frac{128\gamma}{\alpha\beta^2}} 
\min\left( \frac{\alpha}{\gamma}, \alpha \right) \\
& \leq \frac{\beta}{128 \log \frac{128\gamma}{\alpha\beta^2}} 
\min\left( \frac{r(G_i) r_{\guess}}{r(i)}, r(G_i) \right) \\
& =  \frac{r(G_i) r_{\guess} \beta}{64 \log \frac{128\gamma}{\alpha\beta^2} \cdot 2 \max(r(i),r_{\guess})} \\
& \leq  \frac{r(G_i) r_{\guess} \beta}{64 \log \frac{128\gamma}{\alpha\beta^2} (r(i)+r_{\guess})} \\
& \leq \frac{r(G_i) r_{\guess} \beta}{8r(S)}.
\end{align*}
Thus by Lemma~\ref{lem:chilow}, chi-squared distance is lower bounded by 
\begin{align*}
\left( \frac{r(G_i) r_{\guess} \beta}{8r(S)} \right)^2
\frac{r(i)r(S)}{4(r(i)+r(S))^2}  
&  = \frac{r^2_{\guess} \beta^2 r(i) r^2(G_i)}{2^{8} r(S) (r(i)+r(S))^2}\\
& \stackrel{(a)}{\geq} \frac{r^2_{\guess} \beta^2 r(i) r^2(G_i)}{2^{20} (r(i)+r_{\guess})^3 \log^3 \frac{128\gamma}{\alpha\beta^2} }\\
& {\geq} \frac{r^2_{\guess} \beta^2 r(i) r^2(G_i)}{2^{20} \cdot 8 \max (r^3(i),r^3_{\guess}) \cdot \log^3 \frac{128\gamma}{\alpha\beta^2} }\\
& = \frac{r^2(G_i)\beta^2}{2^{23} \log^3 \frac{128\gamma}{\alpha\beta^2}} \min\left( \frac{r(i)}{r_{\guess}} , \frac{r^2_\guess}{r^2(i)}\right)\\
& \stackrel{(b)}{\geq}  \frac{r^2(G_i)\beta^2}{2^{23} \log^3 \frac{128\gamma}{\alpha\beta^2}} \min\left( \frac{\beta r(G_i)}{\gamma} , 
\frac{r^2(G_i)}{\gamma^2 r^2(G_i)}\right)\\
& \geq \frac{\alpha^3 \beta^3}{2^{23} \gamma^2\log^3 \frac{128\gamma}{\alpha\beta^2}}.
\end{align*}
$(a)$ follows from Equation~\eqref{eq:aux_lower3} and $(b)$ follows from bounds on $r_{\guess}$.
Thus with probability $\geq 1- \frac{\alpha\beta^2}{128\gamma}$, \textsc{Test-equal} outputs \texttt{diff}. 
By the union bound, the error probability for an $S \in \{S_1,S_2,\ldots S_m\}$ is $\leq 1-  \frac{\alpha\beta^2}{32\gamma} +  \frac{\alpha\beta^2}{128 \gamma} +   \frac{\alpha\beta^2}{128 \gamma} + \delta'  \leq 1 - \frac{\alpha\beta^2}{128\gamma}$. Since we are repeating it for $m$ sets,  the probability that it outputs \texttt{diff} is 
\[
\geq 1 - \left(1 - \frac{\alpha\beta^2}{128\gamma} \right)^m \geq 
1 - e^{ \frac{\alpha \beta^2m}{128\gamma}} \geq  \frac 15.
\]
Theorem~\ref{thm:closeness} follows from Lemma~\ref{lem:same} for the case $p=q$. If $\norm{p-q}_1 \geq \epsilon$, then it follows from ~\ref{lem:find_approx} (finds a good tuple), ~\ref{lem:reaches} (finds good approximation of $r_{\guess}$),~\ref{lem:prune_closeness} (pruning),
and ~\ref{lem:diff_1} ($\Pr(D)$ is large), and~\ref{lem:diff_2} ($\Pr(D)$ is small).
By Lemma~\ref{lem:same} success probability when $p=q$ is $\geq 1 - \delta$.
The success probability when $\norm{p-q}_1 \geq \epsilon$ is at least the probability that we pick a good-tuple $(i,\beta,\alpha)$ and the success probability once a good tuple is picked (sum of errors in Lemmas~\ref{lem:reaches},~\ref{lem:prune_closeness} + maximum of errors in Lemmas~\ref{lem:diff_1} and~\ref{lem:diff_2}) which can be shown to be $1/5 \cdot (1/5-\delta-2\delta/5) \geq1/30 $. We now analyze the number of samples our algorithm uses. 

We first calculate the number of samples used by \textsc{Assisted-closeness-test}. Step $2$ calls  \textsc{Prune-set} $m$ times and each time \textsc{Prune-set} uses $n_1 n_2$ samples. 
Hence, step $2$ uses $m n_1 n_2$ samples. Step $3(a)$ uses $mn_4 \cdot \tcO(\beta''^{-2})$ and step $3(b)$ uses $m\cdot \tcO(\epsilon^{-3})$.
Hence, the total number of samples used by \textsc{Assisted-closeness-test} is 
\[
mn_1n_2 + mn_4 \cdot \tcO(\beta''^{-2}) + m \cdot \tcO(\epsilon^{-3}) = \tcO\left( \alpha^{-1}\beta^{-2}\epsilon^{-1} + \alpha^{-1} \beta^{-2} \epsilon^{-1} \epsilon^{-2} + \alpha^{-1}\beta^{-2} \epsilon^{-3}\right) = \tcO \left(\beta^{-1} \epsilon^{-4} \right).
\]
Thus each \textsc{Assisted-closeness-test} uses $\tcO(\epsilon^{-4} \beta^{-1})$ samples. Hence, the number of samples used by \textsc{Binary-search} is 
\[
\leq \tcO \left(\log \log k \left(n_1n_2 + \epsilon^{-4}\beta^{-1} + n_3  \right) \right) = \tcO \left( \frac{\log \log k }{\epsilon^{4}\beta} \right).
\]
Since \textsc{Closeness-test} calls \textsc{Binary-search} for $16/\epsilon$ different tuples. Hence, the sample complexity of closeness test is 
\[
\frac{16}{\epsilon} + \sum^{16/\epsilon}_{j=1}  \tcO \left( \frac{\log \log k}{\epsilon^{4}\beta_j} \right) = \tcO \left(\frac{\log \log k}{\epsilon^5} \right).
\]
\qedhere

\end{document}